\documentclass[reprint,
superscriptaddress,
 amsmath,amssymb,
 aps,
]{revtex4-1}

\usepackage{ulem}
\usepackage{wrapfig}
\usepackage{algorithm}
\usepackage[noend]{algpseudocode}
\usepackage{braket}
\usepackage{amsthm}
\usepackage{tikz}
\usepackage{xfrac}
\usepackage{graphicx}
\usepackage{dcolumn}
\usepackage{bm}
\usepackage{mathtools}
\usepackage{dsfont}
\usepackage[pdftex,colorlinks=true,linkcolor=blue,citecolor=blue,urlcolor=black]{hyperref}
\usepackage{chngcntr}

\renewcommand\bra[1]{{\langle{#1}|}}
\makeatletter
\renewcommand\ket[1]{%
  \@ifnextchar\bra{\k@t{#1}\!}{\k@t{#1}}%
}
\newcommand\k@t[1]{{|{#1}\rangle}}
\makeatother

\DeclareMathOperator{\tr}{tr}
\newcommand{\Id}{\mathds{1}}

\newtheorem{thm}{Theorem}
\newtheorem*{thm*}{Theorem}

\newtheorem{lem}[thm]{Lemma}
\newtheorem{prop}[thm]{Proposition}

\makeatletter
\newtheorem*{rep@theorem}{\rep@title}
\newcommand{\newreptheorem}[2]{%
\newenvironment{rep#1}[1]{%
 \def\rep@title{#2 \ref{##1} (restated)}%
 \begin{rep@theorem}}%
 {\end{rep@theorem}}}
\makeatother

\newreptheorem{thm}{Theorem}
\newreptheorem{lem}{Lemma}
\newreptheorem{cor}{Corollary}

\begin{document}
\normalem


\title{Optimal verification of entangled states with local measurements}

\author{Sam Pallister}
\email{sam.pallister@bristol.ac.uk}
\affiliation{School of Mathematics, University of Bristol, UK}
\affiliation{Quantum Engineering Centre for Doctoral Training, University of Bristol, UK}
\author{Noah Linden}
\email{n.linden@bristol.ac.uk}
\affiliation{School of Mathematics, University of Bristol, UK}
\author{Ashley Montanaro}
\email{ashley.montanaro@bristol.ac.uk}
\affiliation{School of Mathematics, University of Bristol, UK}

\date{\today}

\begin{abstract}
Consider the task of verifying that a given quantum device, designed to produce a particular entangled state, does indeed produce that state. One natural approach would be to characterise the output state by quantum state tomography; or alternatively to perform some kind of Bell test, tailored to the state of interest.  We show here that neither approach is optimal amongst local verification strategies for two qubit states. We find the optimal strategy in this case and show that quadratically fewer total measurements are needed to verify to within a given fidelity than in published results for quantum state tomography, Bell test, or fidelity estimation protocols. We also give efficient verification protocols for any stabilizer state. Additionally, we show that requiring that the strategy be constructed from local, non-adaptive and non-collective measurements only incurs a constant-factor penalty over a strategy without these restrictions.
\end{abstract}

\maketitle



Efficient and reliable quantum state preparation is a necessary step for all quantum technologies. However, characterisation and verification of such devices is typically a time-consuming and computationally difficult process. For example, tomographic reconstruction of a state of 8 ions required taking $\sim 650,000$ measurements over 10 hours, and a statistical analysis that took far longer~\cite{Haffner2005}; verification of a few-qubit photonic state is similarly challenging~\cite{Carolan2014,Laing2012}. This is also the case in tomography of continuous-variable systems~\cite{Lvovsky2009a,Bellini2012a,Amosov2012a}. One may instead resort to non-tomographic methods to verify that a device reliably outputs a particular state, but such methods typically either: (a) assume that the output state is within some special family of states, for example in compressed sensing~\cite{Flammia2012,Gross2010a} or matrix product state tomography~\cite{Cramer2010}; or (b) extract only partial information about the state, such as when estimating entanglement witnesses~\cite{Toth2005b,Toth2005c}.

Here, we derive the optimal local verification strategy for common entangled states and compare its performance to bounds for non-adaptive quantum state tomography in~\cite{Sugiyama2013} and the fidelity estimation protocol in~\cite{Flammia2011}. Specifically, we demonstrate non-adaptive verification strategies for arbitrary two-qubit states and stabilizer states of $N$ qubits that are constructed from local measurements, and require quadratically fewer copies to verify to within a given fidelity than for these previous protocols. Moreover, the requirement that the measurements be local incurs only a constant factor penalty over the best non-local strategy, even if collective and adaptive measurements are allowed.

\paragraph*{Premise.}

Colloquially, a quantum state verification protocol is a procedure for gaining confidence that the output of some device is a particular state over any other. However, for any scheme involving measurements on a finite number of copies of the output state, one can always find an alternative state within some sufficiently small distance that is guaranteed to fool the verifier. Furthermore, the outcomes of measurements are, in general, probabilistic and a verification protocol collects a finite amount of data; and so any statement about verification can only be made up to some finite statistical confidence. The only meaningful statement to make in this context is the statistical inference that the state output from a device sits within a ball of a certain small radius (given some metric) of the correct state, with some statistical confidence. Thus the outcome of a state verification protocol is a statement like: ``the device outputs copies of a state that has $99\%$ fidelity with the target, with $90\%$ probability''. Note that this is different to the setting of state tomography; a verification protocol answers the question: \textit{``Is the state $\ket{\psi}?$''} rather than the more involved tomographic question: \textit{``Which state do I have?''}. Hence, unlike tomography, a verification protocol may give no information about the true state if the protocol fails.

We now outline the framework for verification protocols that we consider. Take a verifier with access to some set of allowed measurements, and a device that produces states $\sigma_1, \sigma_2, \ldots \sigma_n$ which are supposed to all be $\ket{\psi}$, but may in practice be different from $\ket{\psi}$ or each other. We have the promise that either $\sigma_i = \ket{\psi}\bra{\psi}$ for all $i$, or $\bra{\psi}\sigma_i\ket{\psi}\le 1-\epsilon$ for all $i$. The verifier must determine which is the case with worst-case failure probability $\delta$.

The protocol proceeds as follows. For each $\sigma_i$, the verifier randomly draws a binary-outcome projective measurement $\{P_j,\Id-P_j\}$ from a prespecified set $\mathcal{S}$ with some probability $\mu^i_j$. Label the outcomes ``pass'' and ``fail''; in a ``pass'' instance the verifier continues to state $\sigma_{i+1}$, otherwise the protocol ends and the verifier concludes that the state was not $\ket{\psi}$. If the protocol passes on all $n$ states, then the verifier concludes that the state was $\ket{\psi}$. We impose the constraint that every $P_j \in \mathcal{S}$ \emph{always} accepts when $\sigma_i = \ket{\psi}\bra{\psi}$, $\forall i$ (i.e.\ that $\ket{\psi}$ is in the ``pass'' eigenspace of every projector $P_j \in \mathcal{S}$). This may seem a prohibitively strong constraint, but we later demonstrate that it is both achievable for the sets of states we consider and is always asymptotically favourable to the verifier.

The maximal probability that the verifier passes on copy $i$ is
\begin{equation}
\text{Pr}[\text{Pass on copy }i] = \max_{\substack{\sigma \\ \bra{\psi}\sigma\ket{\psi}\le 1-\epsilon}} \tr(\Omega_i \sigma),
\end{equation}
where $\Omega_i = \sum_j \mu_j^i P_j$. However, the verifier seeks to minimise this quantity for each $\Omega_i$ and hence it suffices to take a fixed set of probabilities and projectors $\{\mu_j, P_j\}$, independent of $i$. Then the verifier-adversary optimisation is
\begin{equation}
\min_\Omega \max_{\substack{\sigma \\ \bra{\psi}\sigma\ket{\psi}\le 1-\epsilon}} \tr(\Omega \sigma) \coloneqq 1 - \Delta_\epsilon,
\end{equation}
where $\Omega = \sum_j \mu_j P_j$. We call $\Omega$ a {\em strategy}. $\Delta_\epsilon$ is the expected probability that the state $\sigma$ fails a single measurement. Then the maximal worst-case probability that the verifier fails to detect that we are in the ``bad'' case that $\bra{\psi}\sigma_i\ket{\psi}\le 1-\epsilon$ for all $i$ is $(1 - \Delta_\epsilon)^n$, so to achieve confidence $1-\delta$ it is sufficient to take
\begin{equation} 
n \ge \frac{\ln \delta^{-1}}{\ln((1-\Delta_{\epsilon})^{-1})} \approx \frac{1}{\Delta_{\epsilon}} \ln \delta^{-1}. 
\end{equation}

Protocols of this form satisfy some useful operational properties:
\begin{enumerate}
	\renewcommand{\theenumi}{\Alph{enumi}}
    \item \emph{Non-adaptivity}. The strategy is fixed from the outset and depends only on the mathematical description of $\ket{\psi}$, rather than the choices of any prior measurements or their measurement outcomes.
    \item \emph{Future-proofing}. The strategy is independent of the infidelity $\epsilon$, and gives a viable strategy for any choice of $\epsilon$. Thus an experimentalist is able to arbitrarily decrease the infidelity $\epsilon$ within which verification succeeds by simply taking more total measurements following the strategy prescription, rather than modifying the prescription itself. The experimentalist is free to choose an arbitrary $\epsilon > 0$ and be guaranteed that the strategy still works in verifying $\ket{\psi}$.
\end{enumerate}

One may consider more general non-adaptive verification protocols given $\mathcal{S}$ and $\{\sigma_i\}$, where measurements do not output ``pass'' with certainty given input $\ket{\psi}$, and the overall determination of whether to accept or reject is based on a more complicated estimator built from the relative frequency of ``pass'' and ``fail'' outcomes. However, we show in the Supplemental Material that these strategies require, asymptotically, quadratically more measurements in $\epsilon$ than those where $\ket{\psi}$ is always accepted. We will also see that the protocol outlined above achieves the same scaling with $\epsilon$ and $\delta$ as the globally optimal strategy, up to a constant factor, and so any other strategy (even based on non-local, adaptive or collective measurements) would yield only at most constant-factor improvements.

Given no constraints on the verifier's measurement prescription, the optimal strategy is to just project on to $\ket{\psi}$. In this case, the fewest number of measurements needed to verify to confidence $1-\delta$ and fidelity $1-\epsilon$ is $n_{opt} = \frac{-1}{\ln\left(1-\epsilon\right)}\ln\frac{1}{\delta} \approx \frac{1}{\epsilon}\ln\frac{1}{\delta}$ (see the Supplemental Material). However, in general the projector $\ket{\psi}\bra{\psi}$ will be non-local, which has the disadvantage of being harder to implement experimentally. This is particularly problematic in quantum optics, for example, where deterministic, unambiguous discrimination of a complete set of Bell states is impossible~\cite{Vaidman1999a, Calsamiglia2001, Ewert2014}. Thus, for each copy there is a fixed probability of the measurement returning a ``null'' outcome; hence, regardless of the optimality of the verification strategy, merely the probability of its successful operation decreases exponentially with the number of measurements. Instead, we seek optimal measurement strategies that satisfy some natural properties that make them both physically realisable and useful to a real-world verifier. We impose the following properties:

\begin{enumerate}
    \item \emph{Locality}. $\mathcal{S}$ contains only measurements corresponding to local observables, acting on a single copy of the output state.
    \item \emph{Projective measurement}. $\mathcal{S}$ contains only binary-outcome, projective measurements, rather than more elaborate POVMs.
    \item \emph{Trust}. The physical operation of each measurement device is faithful to its mathematical description; it behaves as expected, without experimental error.
\end{enumerate}

Thus for multipartite states we only consider strategies where each party locally performs a projective measurement on a single copy, and the parties accept or reject based on their collective measurement outcomes. We also highlight the trust requirement to distinguish from self-testing protocols~\cite{Mayers2004,McKague2012,Yang2013}.

Given this prescription and the set of physically-motivated restrictions, we now derive the optimal verification strategy for some important classes of states. To illustrate our approach, we start with the case of a Bell state before generalising to larger classes of states.


\paragraph*{Bell state verification.}

Consider the case of verifying the Bell state $\ket{\Phi^+}= \frac{1}{\sqrt{2}}(\ket{00}+\ket{11})$. If we maintain a strategy where all measurements accept $\ket{\Phi^+}$ with certainty, then it must be the case that $\Omega\ket{\Phi^+} = \ket{\Phi^+}$. The optimisation problem for the verifier-adversary pair is then given by $\Delta_\epsilon$:
\begin{equation}
\Delta_\epsilon = \max_\Omega \min_{\substack{\sigma \\ \bra{\psi}\sigma\ket{\psi}\le 1-\epsilon}} \tr[\Omega(\ket{\Phi^+}\bra{\Phi^+}-\sigma)].
\end{equation}
However, we show in the Supplemental Material that it is never beneficial for the adversary to: (a) choose a non-pure $\sigma$; or (b) to pick a $\sigma$ such that $\bra{\psi}\sigma\ket{\psi} < 1- \epsilon$. Rewrite $\sigma = \ket{\psi_\epsilon}\bra{\psi_\epsilon}$, where $\ket{\psi_\epsilon} = \sqrt{1-\epsilon}\ket{\Phi^+}+\sqrt{\epsilon}\ket{\psi^\bot}$ for some state $\ket{\psi^\bot}$ such that $\braket{\Phi^+|\psi^\bot}=0$. Then,
\begin{align}
\nonumber \Delta_\epsilon &= \max_\Omega \min_{\ket{\psi^\bot}} \epsilon (\bra{\Phi^+}\Omega\ket{\Phi^+}-\bra{\psi^\bot}\Omega\ket{\psi^\bot})\\
&- 2\sqrt{\epsilon(1-\epsilon)}\text{Re}\bra{\Phi^+}\Omega\ket{\psi^\bot}.
\end{align}

Given that $\Omega\ket{\Phi^+}=\ket{\Phi^+}$, we can simplify by noting that $\bra{\Phi^+}\Omega\ket{\Phi^+}=1$ and $\bra{\Phi^+}\Omega\ket{\psi^\bot}=0$. Thus,
\begin{align}
    \nonumber \Delta_\epsilon &= \max_\Omega \min_{\ket{\psi^\bot}} \epsilon(1-\bra{\psi^\bot}\Omega\ket{\psi^\bot})\\
    &= \epsilon(1-\min_\Omega \max_{\ket{\psi^\bot}} \bra{\psi^\bot}\Omega\ket{\psi^\bot}),
\end{align}
where the verifier controls $\Omega$ and the adversary controls $\ket{\psi^\bot}$. Given that $\ket{\Phi^+}$ is itself an eigenstate of $\Omega$, the worst-case scenario for the verifier is for the adversary to choose $\ket{\psi^\bot}$ as the eigenstate of $\Omega$ with the next largest eigenvalue. If we diagonalise $\Omega$ we can write $\Omega = \ket{\Phi^+}\bra{\Phi^+} + \sum_{j=1}^3 \nu_j \ket{\psi^\bot_j}\bra{\psi^\bot_j}$, where $\braket{\Phi^+|\psi^\bot_j}=0 \; \forall j$. The adversary picks the state $\ket{\psi^\bot_{\text{max}}}$ with corresponding eigenvalue $\nu_{\text{max}} = \max_j \nu_j$. Now, consider the trace of $\Omega$: if $\tr(\Omega) < 2$ then the strategy must be a convex combination of local projectors, at least one of which is rank 1. However, the only rank 1 projector that satisfies $P^+\ket{\Phi^+}=\ket{\Phi^+}$ is $P^+=\ket{\Phi^+}\bra{\Phi^+}$, which is non-local; and therefore $\tr(\Omega) \geq 2$. Combining this with the expression for $\Omega$ above gives $\tr(\Omega) = 1+\sum_j \nu_j \geq 2$. It is always beneficial to the verifier to saturate this inequality, as any extra weight on the subspace orthogonal to $\ket{\Phi^+}$ can only increase the chance of being fooled by the adversary. Thus the verifier is left with the optimisation
\begin{equation}
\min \nu_{\text{max}} = \min \max_k \nu_k, \quad \sum_k \nu_k = 1.
\end{equation}
This expression is optimised for $\nu_j = \frac{1}{3}, j=1,2,3$. In this case, $\Omega = \frac{\Id}{3}$ on the subspace orthogonal to the state $\ket{\Phi^+}$. Then we can rewrite $\Omega$ as
\begin{equation}
\Omega = \frac{1}{3}(P^+_{XX}+P^+_{-YY}+P^+_{ZZ}),
\end{equation}
where $P^+_{XX}$ is the projector onto the positive eigensubspace of the tensor product of Pauli matrices $XX$ (and likewise for $-YY$ and $ZZ$). The operational interpretation of this optimal strategy is then explicit: for each copy of the state, the verifier randomly chooses a measurement setting from the set $\{XX,-YY,ZZ\}$ all with probability $\frac{1}{3}$, and accepts only on receipt of outcome ``+1'' on all $n$ measurements. Note that we could expand $\Omega$ differently, for example by conjugating each term in the above expression by any local operator that leaves $\ket{\Phi^+}$ alone; the decomposition above is only one of a family of optimal strategies. As for scaling, we know that $\Delta_\epsilon = \epsilon(1-\nu_{\text{max}}) = \frac{2\epsilon}{3}$, and the number of measurements needed to verify the Bell state $\ket{\Phi^+}$ is then $n_{opt} = \left[\ln\left(\frac{3}{3-2\epsilon}\right)\right]^{-1}\ln{\frac{1}{\delta}} \approx \frac{3}{2\epsilon}\ln\frac{1}{\delta}$. Note that this is only worse than the optimal non-local strategy by a factor of $1.5$.

In comparison, consider instead verifying a Bell state by performing a CHSH test. Then even in the case of trusted measurements, the total number of measurements scales like $O\left(\frac{1}{\epsilon^2}\right)$~\cite{Sugiyama2014}, which is quadratically worse than the case of measuring the stabilizers $\{XX,-YY,ZZ\}$. This suboptimal scaling is shared by the known bounds for non-adaptive quantum state tomography with single-copy measurements in~\cite{Sugiyama2013} and fidelity estimation in~\cite{Flammia2011}. See~\cite{DaSilva2011,Ferrie2016a,Struchalin2016a} for further discussion of this scaling in tomography. Additionally, two-qubit tomography potentially requires five times as many measurement settings. We also note that a similar quadratic improvement was derived in adaptive quantum state tomography in~\cite{Mahler2013}, in the sample-optimal tomographic scheme in~\cite{Haah2016} and in the quantum state certification scheme in~\cite{Badescu2017a}; however, the schemes therein assume access to either non-local or collective measurements.


\paragraph*{Arbitrary states of two qubits.}

The goal is unchanged for other pure states of two qubits: we seek strategies that accept the target state with certainty, and hence achieve the asymptotic advantage outlined for Bell states above. It is not clear a priori that such a strategy exists for general states, in a way that is as straightforward as the previous construction. However, we show that for any two-qubit state not only does such a strategy exist, but we can optimise within the family of allowable strategies and give an analytic expression with optimal constant factors.

We first remark that we can restrict to states of the form $\ket{\psi_\theta} = \sin\theta\ket{00} +\cos\theta\ket{11}$ without loss of generality, as any state is locally equivalent to a state of this form, for some $\theta$. Specifically, given any two qubit state $\ket{\psi}$ with optimal strategy $\Omega_{opt}$, a locally equivalent state $(U \otimes V)\ket{\psi}$ has optimal strategy $(U \otimes V) \Omega_{opt} (U \otimes V)^\dagger$. The proof of this statement can be found in the Supplemental Material. Given the restriction to this family of states, we can now write down an optimal verification protocol.

\begin{figure*}
\begin{minipage}[t]{0.48\textwidth}
\includegraphics[width=\linewidth]{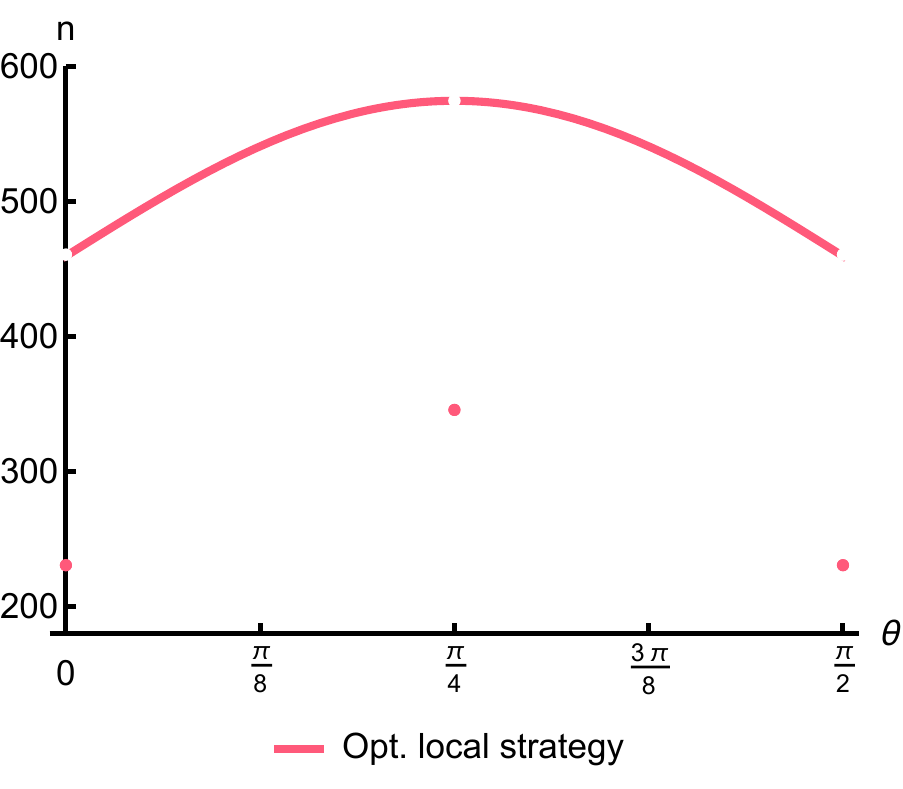}
\caption{The number of measurements needed to verify the state $\Ket{\psi_\theta}=\sin\theta\Ket{00}+\cos\theta\Ket{11}$, as a function of $\theta$, using the optimal strategy. See Eq.~\ref{eq:twoqubitperf}. Here, $1-\epsilon = 0.99$ and $1 - \delta = 0.9$.}
\label{fig:two_qubit_perf}
\end{minipage}\hfill%
\begin{minipage}[t]{0.48\linewidth}
\includegraphics[width=\textwidth]{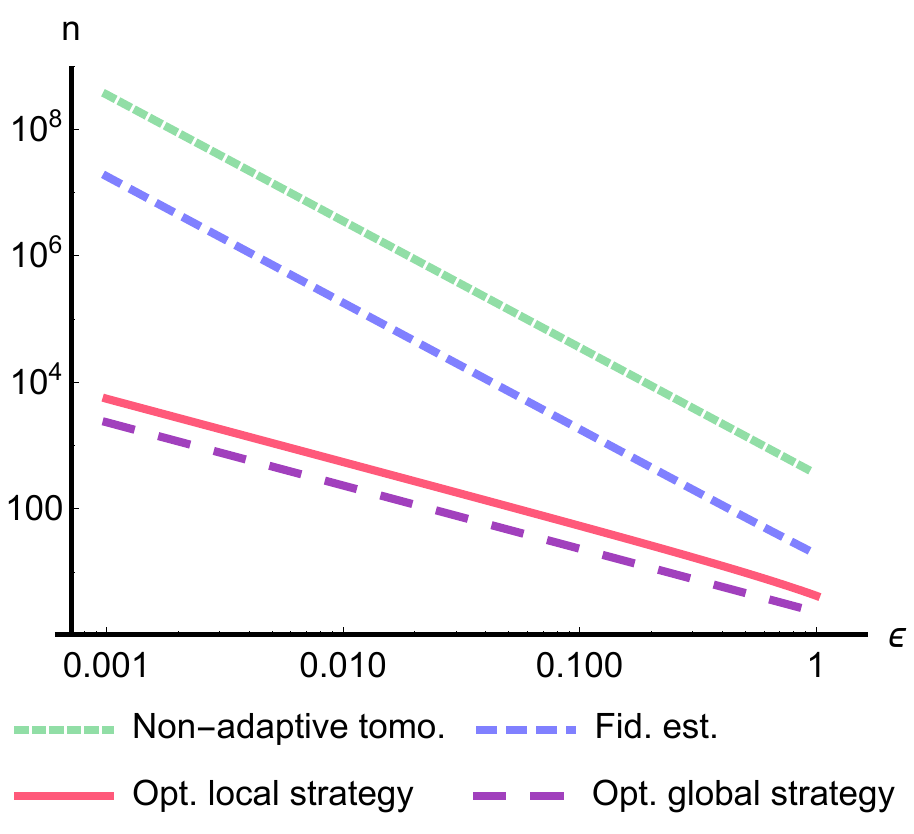}
\caption{A comparison of the total number of measurements required to verify to fidelity $1-\epsilon$ for the strategy derived here, versus the known bounds for estimation up to fidelity $1-\epsilon$ using non-adaptive tomography in~\cite{Sugiyama2013} and the fidelity estimation protocol in~\cite{Flammia2011}, and the globally optimal strategy given by projecting onto $\Ket{\psi}$. Here, $1 - \delta = 0.9$ and $\theta = \frac{\pi}{8}$.}
\label{fig:epsilon_plot}
\end{minipage}%
\end{figure*}
    
\begin{thm}\label{thm:two_qubits}
Any optimal strategy for verifying a state of the form $\ket{\psi_\theta}=\sin\theta\ket{00}+\cos\theta\ket{11}$ for $0 < \theta < \frac{\pi}{2}$, $\theta \neq \frac{\pi}{4}$ that accepts $\ket{\psi_\theta}$ with certainty and satisfies the properties of locality, trust and projective measurement, can be expressed as a strategy involving four measurement settings:
\begin{align}
\nonumber \Omega_{opt} &= \alpha(\theta)P^+_{ZZ}\\
\nonumber &+ \frac{1-\alpha(\theta)}{3}\sum_{k=1}^3 \left[\Id - (\ket{u_k}\otimes\ket{v_k})(\bra{u_k}\otimes\bra{v_k})\right],\\ 
\text{for } &\alpha(\theta) = \frac{2-\sin(2\theta)}{4+\sin(2\theta)},
\end{align}
where $P^+_{ZZ}$ is the projector onto the positive eigenspace of the Pauli operator $ZZ$, and the sets of states $\{\ket{u_k}\}$ and $\{\ket{v_k}\}$ are written explicitly in the Supplemental Material. The number of measurements needed to verify to within infidelity $\epsilon$ and with power $1-\delta$ satisfies
\begin{equation}\label{eq:twoqubitperf}
n_{opt} \approx (2+\sin\theta\cos\theta)\epsilon^{-1}\ln\delta^{-1}.
\end{equation}
\end{thm}

The proof of this theorem is included in the Supplemental Material. Note that the special cases for $\ket{\psi_\theta}$ where $\theta = 0$, $\theta = \frac{\pi}{2}$ and $\theta = \frac{\pi}{4}$ are omitted from this theorem. In these cases, $\ket{\psi_\theta}$ admits a wider choice of measurements that accept with certainty. We have already treated the Bell state case $\theta = \frac{\pi}{4}$ above. In the other two cases, the state $\ket{\psi_\theta}$ is product and hence the globally optimal measurement, just projecting onto $\ket{\psi_\theta}$, is a valid local strategy. We note that this leads to a discontinuity in the number of measurements needed as a function of $\theta$, for fixed $\epsilon$ (as seen in Fig.~\ref{fig:two_qubit_perf}). This arises since our strategies are designed to have the optimal scaling $\left(O\left(\frac{1}{\epsilon}\right)\right)$ for fixed $\theta$, achieved by having strategies that accept $\ket{\psi}$ with probability $1$.

As for scaling, in Fig.~\ref{fig:epsilon_plot} the number of measurements required to verify a particular two-qubit state of this form, for three protocols, is shown. The optimal protocol derived here gives a marked improvement over the previously published bounds for both tomography~\cite{Sugiyama2013} and fidelity estimation~\cite{Flammia2011} for the full range of $\epsilon$, for the given values of $\theta$ and $\delta$. The asymptotic nature of the advantage for the protocol described here implies that the gap between the optimal scheme and tomography only grows as the requirement on $\epsilon$ becomes more stringent. Note also that the optimal local strategy is only marginally worse than the best possible strategy of just projecting onto $\ket{\psi}$.

\paragraph*{Stabilizer states.} Additionally, it is shown in the Supplemental Material that we can construct a strategy with the same asymptotic advantage for any stabilizer state, by drawing measurements from the stabilizer group (where now we only claim optimality up to constant factors). The derivation is analogous to that for the Bell state above, and given that the Bell state is itself a stabilizer state, the strategy above is a special case of the stabilizer strategy discussed below. For a state of $N$ qubits, a viable strategy constructed from stabilizers must consist of at least the $N$ stabilizer generators of $\ket{\psi}$. This is because a set of $k<N$ stabilizers stabilizes a subspace of dimension at least $2^{N-k}$, and so in this case there always exists at least one orthogonal state to $\ket{\psi}$ accessible to the adversary that fools the verifier with certainty. In this minimal case, the number of required measurements is $n_{opt}^{s.g.} \approx N\epsilon^{-1}\ln\delta^{-1}$, with this bound saturated by measuring all stabilizer generators with equal weight. Conversely, constructing a measurement strategy from the full set of $2^N - 1$ linearly independent stabilizers requires a number of measurements $n_{opt}^{stab} \approx \frac{2^N - 1}{ 2^{(N-1)}}\epsilon^{-1}\ln\delta^{-1}$, again with this bound saturated by measuring each stabilizer with equal weight. For growing $N$, the latter expression for the number of measurements is bounded from above by $2\epsilon^{-1}\ln\delta^{-1}$, which implies that there is a local strategy for any stabilizer state, of an arbitrary number of qubits, which requires at most twice as many measurements as the optimal non-local strategy. Note that this strategy may not be exactly optimal; for example, the state $\ket{00}$ is also a stabilizer state, and in this case applying the measurement $\ket{00}\bra{00}$ is both locally implementable and provably optimal. Thus, the exactly optimal strategy may depend more precisely on the structure of the individual state itself. However, the stabilizer strategy is only inferior by a small constant factor. In comparison to the latter strategy constructed from every stabilizer, the former strategy constructed from only the $N$ stabilizer generators of $\ket{\psi}$ has scaling that grows linearly with $N$. Thus there is ultimately a trade-off between number of measurement settings and total number of measurements required to verify within a fixed fidelity.

In principle, the recipe derived here to extract the optimal strategy for a state of two qubits can be applied to any pure state. However, we anticipate that deriving this strategy, including correct constants, may be somewhat involved (both analytically and numerically) for states of greater numbers of qubits.

Following the completion of this work, we became aware of \cite{Dimic2017a} which, among other results, applies a similar protocol to the Bell state verification strategy in the context of entanglement detection.


\begin{acknowledgments}
The authors thank Jeremy Adcock, Sam Morley-Short, Tony Short and Chris Sparrow for helpful discussions, and thank Borivoje Dakic for pointing out \cite{Dimic2017a}. SP was supported by the Bristol Quantum Engineering Centre for Doctoral Training, EPSRC grant EP/L015730/1. AM was supported by EPSRC Early Career Fellowship EP/L021005/1. No new data were created during this study.
\end{acknowledgments}


\bibliographystyle{apsrev4-1}
\bibliography{Verification}

\appendix

\widetext
\clearpage
\begin{center}
\textbf{\large Supplemental Material: Optimal verification of entangled states with local measurements}
\end{center}
\setcounter{equation}{0}
\setcounter{figure}{0}
\setcounter{table}{0}
\setcounter{page}{1}
\makeatletter
\counterwithout{equation}{section}
\renewcommand{\theequation}{S\arabic{equation}}
\renewcommand{\thefigure}{S\arabic{figure}}
\renewcommand{\bibnumfmt}[1]{[S#1]}
\renewcommand{\citenumfont}[1]{S#1}

The contents of the following supplemental material are as follows: in Appendix~\ref{app:verification}, we set up a formal framework for state verification protocols. In Appendix~\ref{app:ver_strat_opt} we simplify the form of the protocol using the set of physically-motivated strategy requirements outlined in the main body. Appendix~\ref{app:two_qubits} is concerned with deriving the optimal strategy for states of two qubits, in particular proving Theorem~1; and in Appendix~\ref{app:stabilizers} we derive efficient verification strategies for stabilizer states. Finally, Appendix~\ref{app:hypothesis_testing} outlines the hypothesis testing framework necessary for this paper. 

\section{Quantum state verification}\label{app:verification}

We first set up a formal framework for general state verification protocols. We assume that we have access to a device $\mathcal{D}$ that is supposed to produce copies of a state $\ket{\psi}$. However, $\mathcal{D}$ might not work correctly, and actually produces (potentially mixed) states $\sigma_1,\sigma_2,\dots$ such that $\sigma_i$ might not be equal to $\ket{\psi}\bra{\psi}$. In order to distinguish this from the case where the device works correctly by making a reasonable number of uses of $\mathcal{D}$, we need to have a promise that these states are sufficiently far from $\ket{\psi}$. So we are led to the following formulation of our task:\\

\noindent Distinguish between the following two cases:
\begin{enumerate}
\renewcommand{\theenumi}{(\alph{enumi})}
\item {\bf (Good)} $\sigma_i = \ket{\psi}\bra{\psi}$ for all $i$;
\item {\bf (Bad)} For some fixed $\epsilon$, $F(\ket{\psi},\sigma_i) := \braket{\psi|\sigma_i|\psi} \le 1-\epsilon$ for all $i$.
\end{enumerate}

Given a verifier with access to a set of available measurements $\mathcal{S}$, the protocols we consider for completing this task are of the following form:

\renewcommand{\thealgorithm}{}
\floatname{algorithm}{Protocol}
\begin{algorithm}[H]
  \caption{Quantum state verification}
  \label{alg:state_verification}
   \begin{algorithmic}[1]
   \For{$i = 1$ to $n$}
   \State Two-outcome measurement $M_i \in \mathcal{S}$ on $\sigma_i$, where $M_i$'s outcomes are associated with ``pass'' and ``fail''
   \If{``fail'' is returned}
   \State Output ``reject''
   \EndIf
   \EndFor
   \State Output ``accept''
   \end{algorithmic}
\end{algorithm}

We impose the conditions that in the good case, the protocol accepts with certainty, whereas in the bad case, the protocol accepts with probability at most $\delta$; we call $1-\delta$ the \emph{statistical power} of the protocol. We then aim to find a protocol that minimises $n$ for a given choice of $\ket{\psi}$, $\epsilon$ and $\mathcal{S}$, such that these constraints are satisfied. Insisting that the protocol accepts in the good case with certainty implies that all measurements in $\mathcal{S}$ are guaranteed to pass in this case. This is a desirable property in itself, but one could consider more general non-adaptive protocols where measurements do not output ``pass'' with certainty on $\ket{\psi}$, and the protocol determines whether to accept based on an estimator constructed from the relative frequency of ``pass'' and ``fail'' outcomes across all $n$ copies. We show in Appendix~\ref{app:hypothesis_testing} that this class of protocols has quadratically worse scaling in $\epsilon$ than protocols where each measurement passes with certainty on $\ket{\psi}$.

We make the following observations about this framework:

\begin{enumerate}
\item Given no restrictions on $M_i$, the optimal protocol is simply for each measurement to project onto $\ket{\psi}$. In fact, this remains optimal even over the class of more general protocols making use of adaptivity or collective measurements. One can see this as follows: if a two-outcome measurement $M$ (corresponding to the whole protocol) is described by measurement operators $P$ (accept) and $I-P$ (reject), then if $M$ accepts $\ket{\psi}^{\otimes n}$ with certainty, we must have $P = \ket{\psi}\bra{\psi}^{\otimes n} + P'$ for some residual positive semidefinite operator $P'$. Then replacing $P$ with $\ket{\psi}\bra{\psi}^{\otimes n}$ gives at least as good a protocol, as the probability of accepting $\ket{\psi}$ remains 1, while the probability of accepting other states cannot increase.

The probability of acceptance in the bad case after $n$ trials is then at most $(1-\epsilon)^n$, so it is sufficient to take
\begin{equation} \label{eq:globalopt} n \ge \frac{\ln \delta^{-1}}{\ln((1-\epsilon)^{-1})} \approx \epsilon^{-1} \ln \delta^{-1} \end{equation}
to achieve statistical power $1-\delta$. This will be the yardstick against which we will compare our more restricted protocols below.

\item We assume that the states $\sigma_i$ are independently and adversarially chosen. This implies that if (as we will consider below) $\mathcal{S}$ contains only projective measurements and does not contain the measurement projecting onto $\ket{\psi}\bra{\psi}$, it is necessary to choose the measurement $M_i$ at random from $\mathcal{S}$ and unknown to the adversary. Otherwise, we could be fooled with certainty by the adversary choosing $\sigma_i$ to have support only in the ``pass'' eigenspace of $M_i$ for each copy $i$.

\item We can be explicit about the optimisation needed to derive the optimal protocol in this adversarial setting. As protocols of the above form reject whenever a measurement fails, the adversary's goal at the $i$'th step is to maximise the probability that the measurement $M_i$ at that step passes on $\sigma_i$. If the $j$'th measurement setting in $\mathcal{S}$, $M^j$, is picked from $\mathcal{S}$ at step $i$ with probability $\mu_j^i$, the largest possible overall probability of passing for copy $i$ is
\begin{equation}
\text{Pr}[\text{Pass on copy }i] = \max_{\sigma_i,\braket{\psi|\sigma_i|\psi} \le 1-\epsilon} \sum_j \mu^i_j \tr(P_j \sigma_i),
\end{equation}
where we denote the corresponding ``pass'' projectors $P_j$. We can write $\Omega_i = \sum_j \mu_j^i P_j$, and then
\begin{equation}
\text{Pr}[\text{Pass on copy }i] = \max_{\sigma,\braket{\psi|\sigma|\psi} \le 1-\epsilon} \tr(\Omega_i \sigma).
\end{equation}
As the verifier, we wish to minimise this expression over all $\Omega_i$, so we end up with a final expression that does not depend on $i$. This leads us to infer that optimal protocols of this form can be assumed to be non-adaptive in two senses: they do not depend on the outcome of previous measurements (which is clear, as the protocol rejects if it ever sees a ``fail'' outcome); and they also do not depend on the measurement choices made previously.

Therefore, in order to find an optimal verification protocol, our task is to determine
\begin{equation} \label{eq:minsigma} \min_\Omega \max_{\sigma,\braket{\psi|\sigma|\psi} \le 1-\epsilon} \tr(\Omega \sigma), \end{equation}
where $\Omega$ is an operator of the form $\Omega = \sum_j \mu_j P_j$ for $P_j \in \mathcal{S}$ and some probability $\mu_j$. We call such operators {\em strategies}. 
If $\mathcal{S}$ contained all measurement operators (or even all projectors), $\Omega$ would be an arbitrary operator satisfying $0 \le \Omega \le I$. However, this notion becomes nontrivial when one considers restrictions on $\mathcal{S}$. Here, we focus on the experimentally motivated case where $\mathcal{S}$ contains only projective measurements that can be implemented via local operations and classical postprocessing.

\item In a non-adversarial scenario, it may be acceptable to fix the measurements in $\Omega$ in advance, with appropriate frequencies $\mu_j$. Then, given $n$, a strategy $\Omega = \sum_j \mu_j P_j$ corresponds to a protocol where for each $j$ we deterministically make $\mu_j n$ measurements $\{P_j,I-P_j\}$. For large $n$, and fixed $\sigma_i = \sigma$, this will achieve similar performance to the above protocol.

\item More complicated protocols with adaptive or collective measurements, or measurements with more than two outcomes, cannot markedly improve on the strategies derived here. We do not treat these more general strategies explicitly, but note that the protocols we will describe based on local projective measurements already achieve the globally optimal bound (\ref{eq:globalopt}) up to constant factors, so any gain from these more complex approaches would be minor.
\end{enumerate}

\section{Verification strategy optimisation}\label{app:ver_strat_opt}

In this appendix, we simplify the form of the optimisation in~\ref{eq:minsigma} using the strategy requirements outlined previously. We start by making the following useful observation:

\begin{lem}
We can assume without loss of generality that, in (\ref{eq:minsigma}), $\sigma$ is pure.
\end{lem}
\begin{proof}
Assume the adversary chooses a fixed density matrix $\sigma$, which is globally optimal: it forces the verifier to accept $\sigma$ with the greatest probability among states $\sigma$ such that $\braket{\psi|\sigma|\psi} \coloneqq r \le 1-\epsilon$. The probability of accepting this $\sigma$ given strategy $\Omega$ is then
\begin{equation}
\Pr[\text{Accept } \sigma] = \tr(\Omega\sigma).
\end{equation}
We have asserted that $\Omega$ accepts $\ket{\psi}$ with certainty: $\bra{\psi}\Omega\ket{\psi}=1$. However, for this to be the case $\Omega$ must have $\ket{\psi}$ as an eigenstate with eigenvalue $1$; thus we can write
\begin{equation}
\Omega = \ket{\psi}\bra{\psi} + \sum_j c_j \ket{\psi^\bot_j}\bra{\psi^\bot_j}
\end{equation}
where the states $\{\ket{\psi^\bot_j}\}$ are a set of mutually orthogonal states orthogonal to $\ket{\psi}$. Then
\begin{align}
\Pr[\text{Accept } \sigma] &= \bra{\psi}\sigma\ket{\psi} + \sum_j c_j \bra{\psi^\bot_j}\sigma\ket{\psi^\bot_j}\\
&= r + \sum_j c_j \bra{\psi^\bot_j}\sigma\ket{\psi^\bot_j}.
\end{align}
We can write
\begin{equation}
\sigma = a\ket{\psi}\bra{\psi}+ b\sigma^\bot + c \ket{\psi}\bra{\Phi^\bot} + c^*\ket{\Phi^\bot}\bra{\psi},
\end{equation}
where $\sigma^\bot$ is a density matrix entirely supported in the subspace spanned by the states $\ket{\psi^\bot_j}$, and $\ket{\Phi^\bot}$ is a vector in the subspace spanned by $\ket{\psi^\bot_j}$. We know that $a=r$ as $\bra{\psi}\sigma\ket{\psi} = r$, and $b= 1-r$ as $\tr(\sigma) = 1$. Now, note that the probability of accepting $\sigma$ does not depend on the choice of $\ket{\Phi^\bot}$. Thus $\tr(\Omega\sigma)$ is maximised when $\sigma^\bot = \ket{\psi^\bot_{max}}\bra{\psi^\bot_{max}}$, where $\ket{\psi^\bot_{max}}$ is the orthogonal state in the spectral decomposition of $\Omega$ with largest eigenvalue, $c_{max}$. Thus
\begin{equation}
\max_\sigma \tr(\Omega\sigma) = r + (1-r)c_{max},
\end{equation}
which is achieved by any density matrix of the form
\begin{equation}
\sigma = r\ket{\psi}\bra{\psi}+ (1-r)\ket{\psi^\bot_{max}}\bra{\psi^\bot_{max}} + c \ket{\psi}\bra{\Phi^\bot} + c^*\ket{\Phi^\bot}\bra{\psi}.
\end{equation}
Note that the pure state $\sigma = \ket{\phi}\bra{\phi}$ for $\ket{\phi} = \sqrt{r}\ket{\psi}+\sqrt{1-r}\ket{\psi^\bot_{max}}$ is of this form, and so we can assume that the adversary makes this choice.
\end{proof}

Given that the state $\sigma$ can be taken to be pure and that the fidelity $F(\ket{\psi},\sigma) \le 1- \epsilon$, we write $\sigma = \ket{\psi_{\bar{\epsilon}}}\bra{\psi_{\bar{\epsilon}}}$, where $\ket{\psi_{\bar{\epsilon}}} := \sqrt{1-\bar{\epsilon}}\ket{\psi} +\sqrt{\bar{\epsilon}} \ket{\psi^\bot}$ and $\braket{\psi|\psi^\bot}=0$, for some $\bar{\epsilon}\ge\epsilon$ chosen by the adversary, to be optimised later. Denote 
\begin{equation}
\min_\Omega \max_{\substack{\sigma \\ \bra{\psi}\sigma\ket{\psi}\le 1-\epsilon}} \tr(\Omega \sigma) \coloneqq 1 - \Delta_\epsilon.
\end{equation}
Then the optimisation problem becomes to determine $\Delta_\epsilon$, where
\begin{align}\label{eq:optimisation_arbitrary}
&\Delta_\epsilon = \max_\Omega \min_{\ket{\psi^\bot}, \bar{\epsilon} \geq \epsilon} \bar{\epsilon} (1-\bra{\psi^\bot}\Omega\ket{\psi^\bot}) - 2\sqrt{\bar{\epsilon}(1-\bar{\epsilon})}\text{Re}(\bra{\psi}\Omega\ket{\psi^\bot})\\
\nonumber &\text{and } \Omega\ket{\psi}=\ket{\psi}.
\end{align}
This expression can be simplified given that $\Omega\ket{\psi}=\ket{\psi}$. In particular, we then know that $\bra{\psi^\bot}\Omega\ket{\psi}=0$ for any choice of orthogonal state $\ket{\psi^\bot}$. Thus the term $\sqrt{\bar{\epsilon}(1-\bar{\epsilon})}\text{Re}(\bra{\psi}\Omega\ket{\psi^\bot})$ automatically vanishes. We are then left with the optimisation
\begin{align}
& \Delta_\epsilon = \max_\Omega \min_{\ket{\psi^\bot}, \bar{\epsilon} \geq \epsilon} \bar{\epsilon} (1-\bra{\psi^\bot}\Omega\ket{\psi^\bot}),\\
\nonumber &\text{where } \Omega\ket{\psi}=\ket{\psi}.
\end{align}
As for the optimisation of $\bar{\epsilon}$, note that it is the goal of the adversary to make $\Delta_\epsilon$ as small as possible; and so they are obliged to set $\bar{\epsilon}=\epsilon$. Then the optimisation becomes
\begin{align}\label{eq:final_optimisation}
\Delta_\epsilon = \epsilon &\max_\Omega \min_{\ket{\psi^\bot}} (1-\bra{\psi^\bot}\Omega\ket{\psi^\bot}),\\
\nonumber &\text{where } \Omega\ket{\psi}=\ket{\psi}.
\end{align}

Note that this expression implies that any $\Omega$ where $\Omega\ket{\psi}=\ket{\psi}$ automatically satisfies the \emph{future-proofing} property: firstly that $\Omega$ is independent of $\epsilon$, but also that the strategy must be viable for any choice of $\epsilon$ (i.e. there must not be a choice of $\epsilon$ where $\Delta_{\epsilon}=0$). For an initial choice $\Delta_{\epsilon}>0$, we have that $1-\bra{\psi^\bot}\Omega\ket{\psi^\bot}>0$ and so $\Delta_{\epsilon'}>0$ for any $0 < \epsilon' < \epsilon$. Thus the verifier is free to decrease $\epsilon$ arbitrarily without fear of the strategy failing. Note also that this condition may not be automatically guaranteed if the verifier chooses an $\Omega$ such that $\Omega\ket{\psi} \neq \ket{\psi}$.

Regarding the optimisation problem in~\ref{eq:final_optimisation}, for an arbitrary state $\ket{\psi}$ on $n$ qubits it is far from clear how to: (a) construct families of viable $\Omega$ (built from local projective measurements) that accept $\ket{\psi}$ with certainty; (b) to then solve this optimisation problem over those families of $\Omega$. For the remainder of this work, we focus on states of particular experimental interest where we can solve the problem: arbitrary states of 2 qubits, and stabilizer states.


\section{States of two qubits}\label{app:two_qubits}

We now derive the optimal verification strategy for an arbitrary pure state of two qubits. We first give the proof of the statement in the main text that optimal strategies for locally equivalent states are easily derived by conjugating the strategy with the local map that takes one state to the other. Hence, we can restrict our consideration to verifying states of the form $\ket{\psi}=\sin\theta\ket{00}+\cos\theta\ket{11}$ without loss of generality. Specifically:

\begin{lem}
Given any two qubit state $\ket{\psi}$ with optimal strategy $\Omega_{opt}$, a locally equivalent state $(U \otimes V)\ket{\psi}$ has optimal strategy $(U \otimes V) \Omega_{opt} (U \otimes V)^\dagger$.
\end{lem}

\begin{proof}
We must show that strategy $\Omega' = (U \otimes V) \Omega_{opt} (U \otimes V)^\dagger$ is both a valid strategy, and is optimal for verifying $\ket{\psi'}=(U \otimes V)\ket{\psi}$.

\noindent \emph{Validity}: If $\Omega_{opt} = \sum_j \mu_j P_j$ is a convex combination of local projectors, then so is $\Omega'$:
\begin{align}
\nonumber \Omega' = (U \otimes V)\Omega(U \otimes V)^\dagger &= \sum_j \mu_j (U \otimes V)P_j(U \otimes V)^\dagger\\ &= \sum_j \mu_j P'_j.
\end{align}
Also, if $\Omega_{opt}\ket{\psi}=\ket{\psi}$  then $\Omega'\ket{\psi'}=\ket{\psi'}$:
\begin{align}
\Omega_{opt}\ket{\psi} = \ket{\psi} &\Rightarrow (U \otimes V)\Omega\ket{\psi}=p_{opt}(U \otimes V)\ket{\psi}\\
\nonumber &\Rightarrow (U \otimes V)\Omega(U \otimes V)^\dagger(U \otimes V)\ket{\psi}=(U \otimes V)\ket{\psi}\\
\nonumber &\Rightarrow \Omega' \ket{\psi'} = \ket{\psi'}.
\end{align}
\emph{Optimality}:
The performance of a strategy is determined by the maximum probability of accepting an orthogonal state $\ket{\psi^\bot}$. For the strategy-state pairs $(\Omega_{opt},\ket{\psi})$ and $(\Omega',\ket{\psi'})$, we denote this parameter $q_{opt}$ and $q'$, respectively. Then
\begin{align}
    q_{opt} &= \max_{\ket{\psi^\bot}} \bra{\psi^\bot} \Omega_{opt} \ket{\psi^\bot} = \max_{\ket{\phi}, \braket{\psi | \phi} = 0} \bra{\phi} \Omega_{opt} \ket{\phi}\\
    &= \max_{(U \otimes V)\ket{\phi}, \bra{\psi}(U \otimes V)^\dagger(U \otimes V)\ket{\phi} = 0} \bra{\phi}(U \otimes V)^\dagger (U \otimes V) \Omega_{opt} (U \otimes V)^\dagger (U \otimes V)\ket{\phi}\\
    &= \max_{\ket{\phi'}, \braket{\psi' | \phi'} = 0} \bra{\phi'} \Omega' \ket{\phi'} = q'.
\end{align}
So applying the same local rotation to the strategy and the state results in no change in the performance of the strategy. Thus the following simple proof by contradiction holds: assume that there is a better strategy for verifying $\ket{\psi'}$, denoted $\Omega''$. But then the strategy $(U \otimes V)^\dagger \Omega'' (U \otimes V)$ must have a better performance for verifying $\ket{\psi}$ than $\Omega_{opt}$, which is a contradiction. Thus $\Omega'$ must be the optimal strategy for verifying $\ket{\psi'}$.
\end{proof}

\noindent We will now prove Theorem~1 from the main body. However, we first prove a useful lemma - that no optimal strategy can contain the identity measurement (where the verifier always accepts regardless of the tested state). In the following discussion, we denote the projector $\Pi \coloneqq \Id - \ket{\psi}\bra{\psi}$. For a strategy $\Omega$ where $\Omega\ket{\psi}=\ket{\psi}$, the quantity of interest which determines $\Delta_\epsilon$ in (\ref{eq:final_optimisation}) is the maximum probability of accepting an orthogonal state $\ket{\psi^\bot}$:
\begin{equation}
q \coloneqq \Vert \Pi \Omega \Pi \Vert = \max_{\ket{\psi^\bot}} \bra{\psi^\bot}\Omega\ket{\psi^\bot}. 
\end{equation}
If a strategy is augmented with an accent or subscript, the parameter $q$ inherits that accent or subscript.


\begin{lem}\label{lem:no_id_part}
Consider an operator $0 \leq \Omega \leq 1$, $\Omega\ket{\psi}=\ket{\psi}$ of the form $\Omega = (1-\alpha)\Omega_1 + \alpha \Id$ for $0 \leq \alpha \leq 1$. Then $q \geq q_1$.  
\end{lem}
\begin{proof}
For arbitrary $\ket{\psi^\perp}$ such that $\braket{\psi|\psi^\perp}=0$, $\braket{\psi^\perp|\Omega|\psi^\perp} = (1-\alpha) \braket{\psi^\perp|\Omega_1|\psi^\perp} + \alpha$. This is maximised by choosing $\ket{\psi^\perp}$ such that $\braket{\psi^\perp|\Omega_1|\psi^\perp} = q_1$, giving $q = (1-\alpha) q_1 + \alpha \ge q_1$.
\end{proof}


\noindent We are now in a position to prove Theorem~1. Note that the special cases where $\ket{\psi}$ is a product state ($\theta = 0$ or $\frac{\pi}{2}$) or a Bell state ($\theta = \frac{\pi}{4}$) are treated separately.

\begin{repthm}{thm:two_qubits}
Any optimal strategy for verifying a state of the form $\ket{\psi}=\sin\theta\ket{00}+\cos\theta\ket{11}$ for $0 < \theta < \frac{\pi}{2}$, $\theta \neq \frac{\pi}{4}$ that accepts $\ket{\psi_\theta}$ with certainty and satisfies the properties of locality, trust and projective measurement, can be expressed as a strategy involving four measurement settings:
\begin{equation}
\Omega^{opt} = \frac{2-\sin(2\theta)}{4+\sin(2\theta)}P^+_{ZZ} + \frac{2(1+\sin(2\theta))}{3(4+\sin(2\theta))}\sum_{k=1}^3 (\Id - \ket{\phi_k}\bra{\phi_k}),
\end{equation}
where the states $\ket{\phi_k}$ are
\begin{align}
\ket{\phi_1} &= \left(\frac{1}{\sqrt{1+\tan\theta}}\ket{0} + \frac{e^{\frac{2\pi i}{3}}}{\sqrt{1+\cot\theta}}\ket{1} \right) \otimes \left(\frac{1}{\sqrt{1+\tan\theta}}\ket{0} + \frac{e^{\frac{\pi i}{3}}}{\sqrt{1+\cot\theta}}\ket{1} \right),\\
\ket{\phi_2} &= \left(\frac{1}{\sqrt{1+\tan\theta}}\ket{0} + \frac{e^{\frac{4\pi i}{3}}}{\sqrt{1+\cot\theta}}\ket{1} \right) \otimes \left(\frac{1}{\sqrt{1+\tan\theta}}\ket{0} + \frac{e^{\frac{5\pi i}{3}}}{\sqrt{1+\cot\theta}}\ket{1} \right),\\
\ket{\phi_3} &= \left(\frac{1}{\sqrt{1+\tan\theta}}\ket{0} + \frac{1}{\sqrt{1+\cot\theta}}\ket{1} \right) \otimes \left(\frac{1}{\sqrt{1+\tan\theta}}\ket{0} - \frac{1}{\sqrt{1+\cot\theta}}\ket{1} \right).
\end{align}
The number of measurements needed to verify to within fidelity $\epsilon$ and statistical power $1-\delta$ is
\begin{equation}
n_{opt} \approx (2+\sin\theta\cos\theta)\epsilon^{-1}\ln\delta^{-1}.
\end{equation}
\end{repthm}

\begin{proof}

The strategy $\Omega$ can be written as a convex combination of local projectors. We can group the projectors by their action according to two local parties, Alice and Bob, and then it must be expressible as a convex combination of five types of terms, grouped by trace:

\begin{equation}\label{eq:strat_general}
    \Omega = c_1 \sum_i \mu_i (\rho^i_1 \otimes \sigma^i_1) + c_2 \sum_j \nu_j (\rho^j_2 \otimes \sigma^j_2 + \rho_2^{j\bot} \otimes \sigma_2^{j\bot}) + c_3\sum_k \eta_k (\Id - \rho^k_3 \otimes \sigma^k_3)+c_4 \sum_l [\zeta_l (\rho^l_4 \otimes \Id) + \xi_l(\Id \otimes \sigma^l_4)]+c_5\Id \otimes \Id,
\end{equation}
where $\rho^k_i$ and $\sigma^k_i$ are single-qubit pure states and the subscript denotes the type of term in question. The state $\rho^{j\bot}$ is the density matrix defined by $\tr(\rho^j \rho^{j\bot})=0$. Qualitatively, given two local parties Alice and Bob with access to one qubit each, and projectors with outcomes $\{\lambda,\bar{\lambda}\}$, the terms above correspond to the following strategies: (1) Alice and Bob both apply a projective measurement and accept if both outcomes are $\lambda$; (2) Alice and Bob both apply a projective measurement and accept if both outcomes agree; (3) Alice and Bob both apply a projective measurement and accept unless both outcomes are $\lambda$; (4) Alice or Bob applies a projective measurement and accepts on outcome $\lambda$, and the other party abstains; and (5) both Alice and Bob accept without applying a measurement. 

We show in Appendix~\ref{app:hypothesis_testing} that strategies that accept $\ket{\psi}$ with certainty have a quadratic advantage in scaling in terms of epsilon. Given this, we enforce this constraint from the outset and then show that a viable strategy can still be constructed. For the general strategy in Eq.~\ref{eq:strat_general} to accept $\ket{\psi}$ with certainty, each term in its expansion must accept $\ket{\psi}$ with certainty. However, this is impossible to achieve for some of the terms in the above expansion. In particular, we show that the terms $(\rho \otimes \sigma)$, $(\rho \otimes \Id)$ and $(\Id \otimes \sigma)$ cannot accept $\ket{\psi}$ with certainty, and the form of the term  $(\rho \otimes \sigma + \rho^\bot \otimes \sigma^\bot)$ is restricted.

$\mathit{(\rho \otimes \sigma)}$: given that $\rho$ and $\sigma$ are pure, write $\rho \otimes \sigma = \ket{u}\bra{u} \otimes \ket{v}\bra{v}$, and so this term only accepts $\ket{\psi}$ with certainty if $\Vert (\ket{u}\bra{u} \otimes \ket{v}\bra{v}) \ket{\psi} \Vert = 1$. However, for $0 < \theta < \frac{\pi}{2}$ the state $\ket{\psi}$ is entangled and this condition cannot be satisfied.

$\mathit{(\rho \otimes \Id)}$ or $\mathit{(\Id \otimes \sigma)}$: For the term $(\rho \otimes \Id)$, reexpress $\rho$ in terms of its Pauli expansion: $\rho \otimes \Id = \frac{1}{2}(\Id+\alpha X + \beta Y + \gamma Z) \otimes \Id$, for $-1 \leq \alpha,\beta,\gamma \leq 1$. Then the condition that this term accepts with probability $p=1$ is
\begin{equation}
\bra{\psi} \frac{1}{2}(\Id+\alpha X + \beta Y + \gamma Z) \otimes \Id\ket{\psi} = 1.
\end{equation}
By inserting the definition of $\ket{\psi}$, this becomes $\frac{1}{2}(1-\gamma\cos(2\theta)) = 1$, which is unsatisfiable for $0<\theta<\frac{\pi}{2}$. It is readily checkable that an identical condition is derived for the term $\Id \otimes \sigma$, given the symmetry of the state $\ket{\psi}$ under swapping.

$\mathit{(\rho \otimes \sigma + \rho^\bot \otimes \sigma^\bot)}$: for this term, we can expand both $\rho$ and $\sigma$ in terms of Pauli operators:
\begin{alignat}{2}
\rho & = \frac{1}{2}(\Id + \alpha X + \beta Y + \gamma Z); \quad &&\rho^\bot = \frac{1}{2}(\Id - \alpha X - \beta Y - \gamma Z)\\
\sigma & = \frac{1}{2}(\Id + \alpha' X + \beta' Y + \gamma' Z); \quad &&\sigma^\bot = \frac{1}{2}(\Id - \alpha' X - \beta' Y - \gamma' Z).
\end{alignat}
Inserting these expressions and the definition of $\ket{\psi}$ into the condition that $p=1$ gives the constraint
\begin{equation}\label{eq:rank2_p=1}
\gamma\gamma' +(\alpha\alpha' - \beta\beta')\sin(2\theta) = 1.
\end{equation}
Now, we know from the Cauchy-Schwarz inequality that
\begin{equation}
\gamma\gamma' +(\alpha\alpha' - \beta\beta')\sin(2\theta) \leq \sqrt{\alpha'^2 + \beta'^2 +\gamma'^2}\sqrt{\alpha^2\sin^2(2\theta) + \beta^2\sin^2(2\theta) +\gamma^2} \leq 1,
\end{equation}
where the second inequality is derived from the fact that $\{\alpha,\beta,\gamma\}$, $\{\alpha',\beta',\gamma'\}$ are the parameterisation of a pair of density matrices. There are two ways that this inequality can be saturated: (a) $\sin(2\theta)=1$; (b) $\alpha\alpha'-\beta\beta'=0$, $\gamma\gamma' = 1$. In all other cases, the inequality is strict. Thus the constraint in Eq.~\ref{eq:rank2_p=1} cannot be satisfied in general. Exception (a) corresponds to $\theta = \frac{\pi}{4}$, which is omitted from this proof and treated separately. In exception (b), we have that $\gamma\gamma' = 1$ and so either $\gamma=\gamma'=1$ or $\gamma=\gamma'=-1$. In both cases we have that
\begin{equation}
\rho \otimes \sigma + \rho^\bot \otimes \sigma^\bot =\left( \frac{\Id + Z}{2}\otimes\frac{\Id+Z}{2}\right) + \left(\frac{\Id - Z}{2}\otimes\frac{\Id-Z}{2}\right) = P^+_{ZZ},
\end{equation}
where $P^+_{ZZ}$ is the projector onto the positive eigenspace of $ZZ$. This is the only possible choice for this particular term that accepts $\ket{\psi}$ with certainty.

We can also make use of Lemma~\ref{lem:no_id_part} to remove the term $\Id \otimes \Id$. Given this and the restrictions above from enforcing that $p=1$, the measurement strategy can be written
\begin{equation}
\Omega = \alpha P^+_{ZZ} + (1-\alpha)\sum_k \eta_k (\Id - \rho_k \otimes \sigma_k),
\end{equation}
where $\sum_k \eta_k = 1$ and $0 \le \alpha \le 1$.

We'll try to further narrow down the form of this strategy by \emph{averaging}; i.e. by noting that, as $\ket{\psi}$ is an eigenstate of a matrix $M_\zeta \otimes M_{-\zeta}$ where
\begin{equation}
M_{\zeta} = \begin{pmatrix}1&0\\0&e^{-i\zeta}\end{pmatrix},
\end{equation}
then conjugating the strategy by $M_\zeta \otimes M_{-\zeta}$ and integrating over all possible $\zeta$ cannot make the strategy worse; if we consider an averaged strategy $\langle \Omega \rangle$ such that
\begin{equation}
\langle\Omega\rangle = \frac{1}{2\pi}\int_{-\pi}^{\pi} d\zeta (M_\zeta \otimes M_{-\zeta})\Omega(M_{-\zeta} \otimes M_{\zeta}),
\end{equation}
then necessarily the performance of $\langle\Omega\rangle$ cannot be worse than that of $\Omega$. To see this, note that the averaging procedure does not affect the probability of accepting the state $\ket{\psi}$. However, for each particular value of $\zeta$ the optimisation for the adversary may necessarily lead to different choices for the orthogonal states $\ket{\psi^\bot (\zeta)}$, and so averaging over $\zeta$ cannot be better for the adversary than choosing the optimal $\ket{\psi^\bot}$ at $\zeta = 0$. 

We can also consider discrete symmetries of the state $\ket{\psi}$. In particular, $\ket{\psi}$ is invariant under both swapping the two qubits, and complex conjugation (with respect to the standard basis); by the same argument, averaging over these symmetries (i.e.\ by considering $\Omega' = \frac{1}{2}(\Omega + (\text{SWAP})\Omega(\text{SWAP}^\dagger))$ and $\Omega'' = \frac{1}{2}(\Omega+\Omega^*)$) cannot produce strategies inferior to the original $\Omega$. Therefore we can consider a strategy averaged over these families of symmetries of $\Omega$, without any loss in performance.

This averaging process is useful for three reasons. Firstly, it heavily restricts the number of free parameters in $\Omega$ requiring optimisation. Secondly, it allows us to be explicit about the general form of $\Omega$. Thirdly, the averaging procedures are distributive over addition; and so we can make the replacement
\begin{align}
\nonumber \Omega = \alpha P^+_{ZZ} + (1-\alpha)\sum_k \eta_k (\Id - \rho_k \otimes \sigma_k) \rightarrow  \langle &\alpha P^+_{ZZ} + (1-\alpha)\sum_k \eta_k (\Id - \rho_k \otimes \sigma_k) \rangle\\
=  &\alpha P^+_{ZZ} + (1-\alpha)\sum_k \eta_k \langle \Id - \rho_k \otimes \sigma_k \rangle.
\end{align}
Note that a single term $\Id - \rho_k \otimes \sigma_k$, may, after averaging, be a convex combination of multiple terms of the form $\Id - \rho \otimes \sigma$. To proceed, we will use this averaging procedure to show that it suffices to only include a single, post-averaging term of the form $\langle \Id - \rho_k \otimes \sigma_k \rangle$ in the strategy $\Omega$, and that the resulting operator can be explicitly decomposed into exactly three measurement settings.

Consider a general operator $\Omega$, expressed as a $4 \times 4$ matrix. First, take the discrete symmetries of $\ket{\psi}$. Averaging over complex conjugation in the standard basis implies that the coefficients of $\langle \Omega \rangle$ are real; and averaging over qubit swapping implies that $\langle \Omega \rangle$ is symmetric with respect to swapping of the two qubits. Denote the operator after averaging these discrete symmetries as $\bar{\Omega}$. Then consider averaging over the continuous symmetry of $\ket{\psi}$:
\begin{align}
\langle \Omega \rangle &= \frac{1}{2\pi}\int_{-\pi}^{\pi}d\zeta (M_{\zeta}\otimes M_{-\zeta})\bar{\Omega}(M_{-\zeta} \otimes M_{\zeta})\\
&= \frac{1}{2\pi}\int_{-\pi}^{\pi}d\zeta
\begin{pmatrix}
1&0&0&0 \\ 0&e^{i\zeta}&0&0 \\ 0&0&e^{-i\zeta}&0 \\ 0&0&0&1
\end{pmatrix}
\begin{pmatrix}\omega_{00}&\omega_{01}&\omega_{01}&\omega_{03}\\\omega_{01}&\omega_{11}&\omega_{12}&\omega_{13}\\\omega_{01}&\omega_{12}&\omega_{11}&\omega_{13}\\\omega_{03}&\omega_{13}&\omega_{13}&\omega_{33}\end{pmatrix}
\begin{pmatrix}
1&0&0&0 \\ 0&e^{-i\zeta}&0&0 \\ 0&0&e^{i\zeta}&0 \\ 0&0&0&1
\end{pmatrix}\\
&= \begin{pmatrix}\omega_{00}&0&0&\omega_{03}\\0&\omega_{11}&0&0\\0&0&\omega_{11}&0\\\omega_{03}&0&0&\omega_{33}\end{pmatrix}.
\end{align}
Thus after averaging using the above symmetries of $\ket{\psi}$, $\langle \Omega \rangle$ can be written in the standard basis as
\begin{equation}\label{eq:omega_averaged}
\langle \Omega \rangle = \begin{pmatrix}a&0&0&b\\0&c&0&0\\0&0&c&0\\b&0&0&d\end{pmatrix},
\end{equation}
for $a,b,c,d \in \mathbb{R}$. Enforcing that the strategy accepts $\ket{\psi}$ with certainty yields $\langle\Omega\rangle\ket{\psi} =\ket{\psi}$, or explicitly that

\begin{equation}
\langle \Omega \rangle = \begin{pmatrix}1-b\cot\theta&0&0&b\\0&c&0&0\\0&0&c&0\\b&0&0&1-b\tan\theta\end{pmatrix}.
\end{equation}
The eigensystem of this operator is then completely specified; besides $\ket{\psi}$, it has the following eigenvectors:
\begin{equation}\label{eq:omega_eigenvectors}
\ket{v_1} = \cos\theta\ket{00}-\sin\theta\ket{11}; \quad \ket{v_2} = \ket{01}; \quad \ket{v_3}=\ket{10},
\end{equation}
with corresponding eigenvalues $\lambda_1 = 1 - b\csc\theta\sec\theta$ and $\lambda_2 = \lambda_3 = c$. The maximum probability of accepting a state orthogonal to $\ket{\psi}$, $q$, can then be written
\begin{equation}
q = \Vert \Pi \langle \Omega \rangle \Pi \Vert = \max \{\lambda_1,\lambda_2\},
\end{equation}
where $\Pi = \Id - \ket{\psi}\bra{\psi}$. Therefore, any reasoning about $q$ can be reduced to reasoning about the pair $(\lambda_1,\lambda_2)$. 

Now, we will show that it suffices to only consider a single term of the form $\langle \Id-\rho_k \otimes \sigma_k \rangle$ in the decomposition of $\Omega$. We write a strategy of this form as
\begin{equation}\label{eq:strat_tr3}
\Omega = \alpha P^+_{ZZ} + (1-\alpha)\langle \Id - \rho \otimes \sigma \rangle.
\end{equation}
For the term $\langle \Id - \rho \otimes \sigma \rangle$, we have a constraint on the trace; if we label the eigenvalues for this term as $\lambda_1^{(3)}$ and $\lambda_2^{(3)}$, we have the constraint that $1+\lambda^{(3)}_1 + 2\lambda_2^{(3)} = \tr \langle\Id-\rho\otimes\sigma\rangle = 3 \Rightarrow \lambda_2^{(3)} = 1 - \frac{\lambda_1^{(3)}}{2}$. The locus of points satisfying this constraint is plotted in the $(\lambda_1,\lambda_2)$ plane as the thick black line in~ Fig.~\ref{fig:convhull}. Moreover, we will show that a single term of this form can achieve any valid choice of $\lambda_1^{(3)}$ on this locus (which we defer until we have an explicit parameterisation of terms of this type; see Eq.~\ref{eq:tr3_av}, below).

However, we also have an additional constraint derived from insisting that the strategy remains local. For example, the point $(0,1)$ in the $(\lambda_1,\lambda_2)$ plane represents the strategy $\Omega = \Id - \ket{v_1}\bra{v_1}$, which corresponds to the strategy where the verifier projects onto $\ket{v_1}$ and accepts if the outcome is not $\ket{v_1}$. But this type of measurement is operationally forbidden as $\ket{v_1}$ is entangled.

It can be readily checked that, for an arbitrary $\theta$, it is not possible to cover the full locus in the range $0 \leq \lambda_1 \leq 1$ with a separable strategy; instead, there is a fixed lower bound on $\lambda_1^{(3)}$. To see this, write
\begin{equation}
\langle \Id - \rho \otimes \sigma \rangle = \ket{\psi}\bra{\psi} + \lambda_1^{(3)} \ket{v_1}\bra{v_1}+\frac{2-\lambda_1^{(3)}}{2}(\ket{v_2}\bra{v_2}+\ket{v_3}\bra{v_3}).
\end{equation}
Then, taking just the $\langle \rho \otimes \sigma \rangle$ part and expressing as a matrix in the computational basis gives
\begin{equation}
\langle \rho \otimes \sigma \rangle = \begin{pmatrix}(1-\lambda_1^{(3)})\cos^2\theta&0&0&(\lambda_1^{(3)}-1)\cos\theta\sin\theta\\0&\frac{\lambda_1^{(3)}}{2}&0&0\\0&0&\frac{\lambda_1^{(3)}}{2}&0\\(\lambda_1^{(3)}-1)\cos\theta\sin\theta&0&0&(1-\lambda_1^{(3)})\sin^2\theta\end{pmatrix}.
\end{equation}
To enforce separability it is necessary and sufficient to check positivity under partial transposition, yielding the constraint $\lambda_1^{(3)} - (1-\lambda_1^{(3)})\sin(2\theta) \geq 0$. Simple rearrangement gives a lower bound that must be satisfied for the strategy to remain separable:
\begin{equation}\label{eq:ppt}
\lambda_1^{(3)} \geq \frac{\sin(2\theta)}{1+\sin(2\theta)} \coloneqq \lambda_{LB}.
\end{equation}
This additional locality constraint rules out any point on the black line to the left of the red point in Fig.~\ref{fig:convhull}. The term $P^+_{ZZ}$ has parameters $\lambda_1^{ZZ} = 1$, $\lambda_2^{ZZ}=0$ and so represents a single point in the $(\lambda_1,\lambda_2)$ plane. Thus the parameters $(\lambda_1,\lambda_2)$ for the full strategy $\Omega$ must be represented by a point in the convex hull of the single point representing the $P^+_{ZZ}$ term and the locus of points representing the trace 3 part - i.e. in the unshaded region in Fig.~\ref{fig:convhull}.

We now show that a strategy that includes more trace 3 terms cannot improve on the performance of the strategy above. Write this expanded strategy as
\begin{equation}\label{eq:strat_multi_tr3}
\Omega' = \alpha P^+_{ZZ} + (1-\alpha)\langle \sum_k \eta_k (\Id - \rho_k \otimes \sigma_k) \rangle,
\end{equation}
for $\sum_k \eta_k = 1$. Firstly, we note again that the averaging operations (SWAP, conjugation via $M_{\zeta}$ and complex conjugation in the standard basis) are distributive over addition and so we can make the replacement
\begin{equation}
\Omega' = \alpha P^+_{ZZ} + (1-\alpha) \sum_k \eta_k \langle \Id - \rho_k \otimes \sigma_k \rangle.
\end{equation}

Write the composite term $\sum_k \eta_k \langle \Id - \rho_k \otimes \sigma_k \rangle \coloneqq \Omega_{\text{comp}}$, with parameters $\lambda_1^{\text{comp}}$ and $\lambda_2^{\text{comp}}$. Note that each term in $\Omega_{\text{comp}}$ satisfies both the constraint from the trace and the constraint from PPT in~\ref{eq:ppt}, and hence so does $\Omega_{\text{comp}}$. Now, each operator in this term shares the same eigenbasis (namely, the set of states $\{\ket{v_i}\}$ in~\ref{eq:omega_eigenvectors}). Thus we know that $\lambda_1^{\text{comp}} = \sum_k \eta_k \lambda_{1,k}$, and likewise for $\lambda_2^{\text{comp}}$; i.e. the strategy parameters for this composite term are just a convex combination of those for its constituent parts. A term $\Omega_{\text{comp}}$ is then specified in the $(\lambda_1,\lambda_2)$ plane by a point $\mathcal{P}_{\text{comp}} = (\lambda_1^{\text{comp}},\lambda_2^{\text{comp}}) \in \text{Conv}(\lambda_{1,k},\lambda_{2,k})$ (i.e. the point $\mathcal{P}_{\text{comp}}$ must lie on the thick black line bounding the unshaded region in Fig.~\ref{fig:convhull}).

Thus we know that $\text{Conv}(\Omega') \subseteq \text{Conv}(\Omega)$, and so any strategy writeable in the form~\ref{eq:strat_multi_tr3} can be replaced by a strategy of the form~\ref{eq:strat_tr3} with identical parameters $(\lambda_1,\lambda_2)$, and hence identical performance. Thus, we need only consider strategies of the form
\begin{equation}
\Omega = \alpha P^+_{ZZ} + (1-\alpha) \langle \Id - \rho \otimes \sigma \rangle.
\end{equation}

We can now be explicit about the form of the above strategy. For $\Omega$ to accept $\ket{\psi}$ with certainty, $\rho \otimes \sigma$ must annihilate $\ket{\psi}$ and so we make the replacement $\rho \otimes \sigma = \ket{\tau}\bra{\tau}$, where $\ket{\tau}$ is the most general pure product state that annihilates $\ket{\psi}$. To be explicit about the form of the state $\ket{\tau}$, write a general two-qubit separable state as
\begin{equation}
\ket{\tau} = (\cos\phi \ket{0} + e^{i\eta}\sin\phi \ket{1}) \otimes (\cos\xi \ket{0} + e^{i\zeta}\sin\xi \ket{1}),
\end{equation}
where we take $0 \leq \phi, \xi \leq \frac{\pi}{2}$, without loss of generality. The constraint that this state annihilates $\ket{\psi}= \sin\theta\ket{00}+\cos\theta\ket{11}$ is 
\begin{equation}\label{eq:annil_constraint}
\cos\phi \cos\xi \sin\theta + e^{-i(\eta+\zeta)}\sin\phi \sin\xi \cos\theta = 0.
\end{equation}
If either $\phi=0$ or $\xi=0$, then $\cos\phi\cos\xi\sin\theta = 0$ implying that $\xi=\frac{\pi}{2}$ or $\phi=\frac{\pi}{2}$, respectively. This yields the annihilating states $\ket{\tau} = \ket{01}$ and $\ket{\tau}=\ket{10}$, respectively. If $\phi, \xi \neq 0$ then from the imaginary part of Eq.~\ref{eq:annil_constraint} we find that $e^{-i(\eta+\zeta)} = -1$. Then we can rearrange to give
\begin{equation}
\tan\phi\tan\xi = \tan\theta.
\end{equation}
Using this constraint and the identities
\begin{equation}
\cos\xi = \frac{1}{\sqrt{1+\tan^2 \xi}}; \quad \sin\xi = \frac{\tan\xi}{\sqrt{1+\tan^2 \xi}},
\end{equation}
we can eliminate $\xi$ to yield
\begin{equation}\label{eq:annil_state_phi}
\ket{\tau} = (\cos\phi \ket{0} + e^{i\eta}\sin\phi \ket{1}) \otimes \left(\frac{\tan\phi}{\sqrt{\tan^2 \phi + \tan^2 \theta}}\ket{0} - \frac{e^{-i\eta}\tan\theta}{\sqrt{\tan^2 \phi + \tan^2 \theta}}\ket{1}\right).
\end{equation}
Note that, for $0<\theta<\frac{\pi}{2}$, taking the limits $\phi \rightarrow 0$ and $\phi \rightarrow \frac{\pi}{2}$ we recover the cases $\ket{\tau}=\ket{01}$ and $\ket{\tau}=\ket{10}$, up to irrelevant global phases. Thus we can proceed without loss of generality by assuming that $\rho \otimes \sigma = \ket{\tau}\bra{\tau}$, where $\ket{\tau}$ is given by Eq.~\ref{eq:annil_state_phi}. Averaging over the symmetries of $\ket{\psi}$ outlined above then yields the following expression:
\begin{equation}\label{eq:tr3_av}
\langle \rho \otimes \sigma \rangle = \frac{1}{t^2 \phi + t^2 \theta} \begin{pmatrix}s^2 \phi&0&0&-s^2\phi t\theta\\0&\frac{1}{2}\left(c^2\phi t^2 \theta + s^2 \phi t^2 \phi\right)&0&0\\0&0&\frac{1}{2}\left(c^2\phi t^2 \theta + s^2 \phi t^2 \phi\right)&0\\-s^2\phi t\theta&0&0&s^2\phi t^2\theta\end{pmatrix},
\end{equation}
using the shorthand $s$, $c$, $t$ for $\sin$, $\cos$ and $\tan$, respectively. Given this explicit parameterisation we can extract the eigenvalue $\lambda_1^{(3)}$:
\begin{equation}
\lambda_1^{(3)} = 1-\frac{\sec^2\theta \sin^2\phi}{\tan^2\theta+\tan^2\phi}. 
\end{equation}
It can be shown by simple differentiation w.r.t. $\phi$ that, for fixed $\theta$, this expression has a minimum at $\lambda_1^{(3)}= \lambda_{LB}$. Also, this expression is a continuous function of $\phi$ and therefore can take any value up to its maximum (namely, $1$). Hence a single trace 3 term is enough to achieve any point in the allowable convex hull in Fig.~\ref{fig:convhull}. For convenience we will denote $\tan^2 \phi = P,\; \tan^2 \theta = T$ for $0 \leq P\leq \infty$, $0 < T < \infty$. The explicit form for the whole strategy is then
\begin{equation}
\Omega = \begin{pmatrix}\frac{T+P(P+T+\alpha)}{(1+P)(P+T)}&0&0&\frac{(1-\alpha)P\sqrt{T}}{(1+P)(P+T)}\\0&\frac{(1-\alpha)(T+2P+P^2 + 2PT)}{2(1+P)(P+T)}&0&0\\0&0&\frac{(1-\alpha)(T+2P+P^2 + 2PT)}{2(1+P)(P+T)}&0\\\frac{(1-\alpha)P\sqrt{T}}{(1+P)(P+T)}&0&0&\frac{T+P(1+P+\alpha T)}{(1+P)(P+T)}\end{pmatrix}.
\end{equation}
We now optimise over the two remaining free parameters, $\{\alpha,\phi\}$ (or alternatively, $\{\alpha,P\}$) for fixed $\theta$ (or fixed $T$). This optimisation is rather straightforward from inspection (see Fig.~\ref{fig:alpha_phi_opt}), and the reader may wish to skip to the answer in Eq.~\ref{eq:two_qubit_opt_strat}. However, we include an analytic proof for the sake of completeness. We have shown that it suffices to consider the eigenvalues $\lambda_1$ and $\lambda_2$, given in this case by the expressions
\begin{equation}\label{eq:strat_eigenvalues}
\lambda_1(\alpha,P,T) = 1 - \frac{P(1-\alpha)(1+T)}{(1+P)(P+T)}; \quad \lambda_2(\alpha,P,T) = (1-\alpha)\left[1-\frac{T+P^2}{2(1+P)(P+T)}\right].
\end{equation}
The parameter $q$ is given by the maximum of these two eigenvalues. Note that, if $P=0$, the expression $\lambda_1(\alpha,0,T)=1$ which implies that the adversary can pick a state that the verifier always accepts, and hence the strategy fails. Likewise, taking the limit $\lim_{P\rightarrow \infty} \lambda_1(\alpha,P,T)=1$. Thus we must restrict to the range $0 < P < \infty$ to construct a viable strategy for the verifier. The quantity $q$ is minimised for fixed $T$ when the derivatives with respect to $P$ and $\alpha$ vanish. First, we calculate the derivatives w.r.t. $\alpha$:
\begin{equation}
\frac{\partial \lambda_1}{\partial \alpha} = \frac{P(1+T)}{(1+P)(P+T)}; \quad \frac{\partial \lambda_2}{\partial \alpha} = \frac{-(2P+P^2+T+2PT)}{2(1+P)(P+T)}.
\end{equation}
Given that $P> 0 $ and $T>0$, we have that for any choice of $T$, $\partial_\alpha \lambda_1 > 0 $ and $\partial_\alpha \lambda_2 < 0$. Thus, one of three cases can occur: (a) for a given choice of $T$ and $P$, the lines given by $\lambda_1$ and $\lambda_2$ intersect in the range $0\leq\alpha\leq 1$ and hence there is a valid $\alpha$ such that $q$ is minimised when $\lambda_1 = \lambda_2$; (b) for a given choice of $T$ and $P$, $\lambda_1 > \lambda_2$ in the range $0\leq\alpha\leq1$ and hence $q$ is minimised when $\alpha=0$; (c) for a given choice of $T$ and $P$, $\lambda_1 < \lambda_2$ in the range $0\leq\alpha\leq 1$ and hence $q$ is minimised when $\alpha=1$. However, we note that this final case cannot occur; it suffices to check that $\lambda_1(\alpha=1) > \lambda_2(\alpha=1)$, and from the expressions in~(\ref{eq:strat_eigenvalues}) we have that $\lambda_1(\alpha=1)=1$ and $\lambda_2(\alpha=1)=0$. As a visual aid for the remaining two cases, see Fig.~\ref{fig:alpha_phi_opt}. In case (a),
\begin{equation}
q = \lambda_1 = \lambda_2 = \frac{1}{2}+\frac{1}{2}\left(\frac{T+P^2}{T+P^2 +4P(1+T)}\right).
\end{equation}
In case (b), we have that
\begin{equation}
q = \lambda_1(0,P,T) = \frac{T+P^2}{(1+P)(P+T)}.
\end{equation}
We must also minimise w.r.t. $\phi$; however, we can safely minimise w.r.t. $P$ as $\partial_\phi P > 0$ (unless $\phi=0$, but in this case $q=1$ and the strategy fails). In case (b), we have
\begin{equation}
\frac{\partial q}{\partial P} = \frac{(P^2 - T)(1+T)}{(1+P)^2(P+T)^2}.
\end{equation}
In this case, consider the two points implicitly defined by the constraint $\lambda_1(0,P,T)=\lambda_2(0,P,T)$ (drawn as the black points in Fig.~\ref{fig:alpha_phi_opt}). Denote these points $f^\pm(T)$. It can be readily checked that in case (b), $\partial_P q < 0$ for any $q < f^-(T)$, and $\partial_P q > 0$ for any $q > f^+(T)$. Thus the minimum w.r.t $P$ must occur when $\lambda_1(0,P,T)=\lambda_2(0,P,T)$ and hence we can restrict our attention to case (a) (note Fig.~\ref{fig:alpha_phi_opt}). In this case, $\partial_P q$ becomes
\begin{equation}
\frac{\partial q}{\partial P} = \frac{-2(1+T)(T-P^2)}{[T+4PT+P(4+P)]^2} = 0,
\end{equation}
which implies that $P = \sqrt{T}$. Substituting in the optimal choices for the parameters $\{\alpha,P\}$ and reexpressing solely in terms of $\theta$ gives the optimal strategy
\begin{equation}\label{eq:two_qubit_opt_strat}
\Omega^{opt} = \frac{2-\sin(2\theta)}{4+\sin(2\theta)}P^+_{ZZ} + \frac{2(1+\sin(2\theta))}{4+\sin(2\theta)}\Omega_3^{opt},
\end{equation}
where $\Omega_3^{opt}$ is given by
\begin{equation}
\Omega_3^{opt} = \Id - \frac{1}{(1+t)^2}\begin{pmatrix}1&0&0&-t\\0&t&0&0\\0&0&t&0\\-t&0&0&t^2\end{pmatrix}, \quad t=\tan\theta.
\end{equation}
This strategy accepts an orthogonal state with probability
\begin{equation}
q_{opt} = \frac{2+\sin(2\theta)}{4+\sin(2\theta)},
\end{equation}
implying that the number of measurements needed to verify to within accuracy $\epsilon$ and with statistical power $1-\delta$ under this test is
\begin{equation}
n_{opt} = \frac{\ln \delta^{-1}}{\ln((1- \Delta_{\epsilon})^{-1})} = \frac{\ln \delta^{-1}}{\ln((1- \epsilon(1-q^{opt}))^{-1})} \approx (2+\sin\theta \cos\theta)\epsilon^{-1}\ln\delta^{-1}.
\end{equation}
The final step is to show that the operator $\Omega_3^{opt}$ can be decomposed into a small set of locally implementable, projective measurements. We can do so with a strategy involving only three terms:
\begin{equation}
\Omega_3^{opt} = \frac{1}{3}\left[\sum_{k=1}^3 (\Id - \ket{\phi_k}\bra{\phi_k})\right],
\end{equation}
where the set of separable states $\{\ket{\phi_k}\}$ are the following:
\begin{align}
\ket{\phi_1} &= \left(\frac{1}{\sqrt{1+\tan\theta}}\ket{0} + \frac{e^{\frac{2\pi i}{3}}}{\sqrt{1+\cot\theta}}\ket{1} \right) \otimes \left(\frac{1}{\sqrt{1+\tan\theta}}\ket{0} + \frac{e^{\frac{\pi i}{3}}}{\sqrt{1+\cot\theta}}\ket{1} \right),\\
\ket{\phi_2} &= \left(\frac{1}{\sqrt{1+\tan\theta}}\ket{0} + \frac{e^{\frac{4\pi i}{3}}}{\sqrt{1+\cot\theta}}\ket{1} \right) \otimes \left(\frac{1}{\sqrt{1+\tan\theta}}\ket{0} + \frac{e^{\frac{5\pi i}{3}}}{\sqrt{1+\cot\theta}}\ket{1} \right),\\
\ket{\phi_3} &= \left(\frac{1}{\sqrt{1+\tan\theta}}\ket{0} + \frac{1}{\sqrt{1+\cot\theta}}\ket{1} \right) \otimes \left(\frac{1}{\sqrt{1+\tan\theta}}\ket{0} - \frac{1}{\sqrt{1+\cot\theta}}\ket{1} \right),
\end{align}
which gives a strategy of the required form.
\end{proof}

\begin{figure*}
\begin{minipage}[t]{0.48\textwidth}
\includegraphics[width=\textwidth]{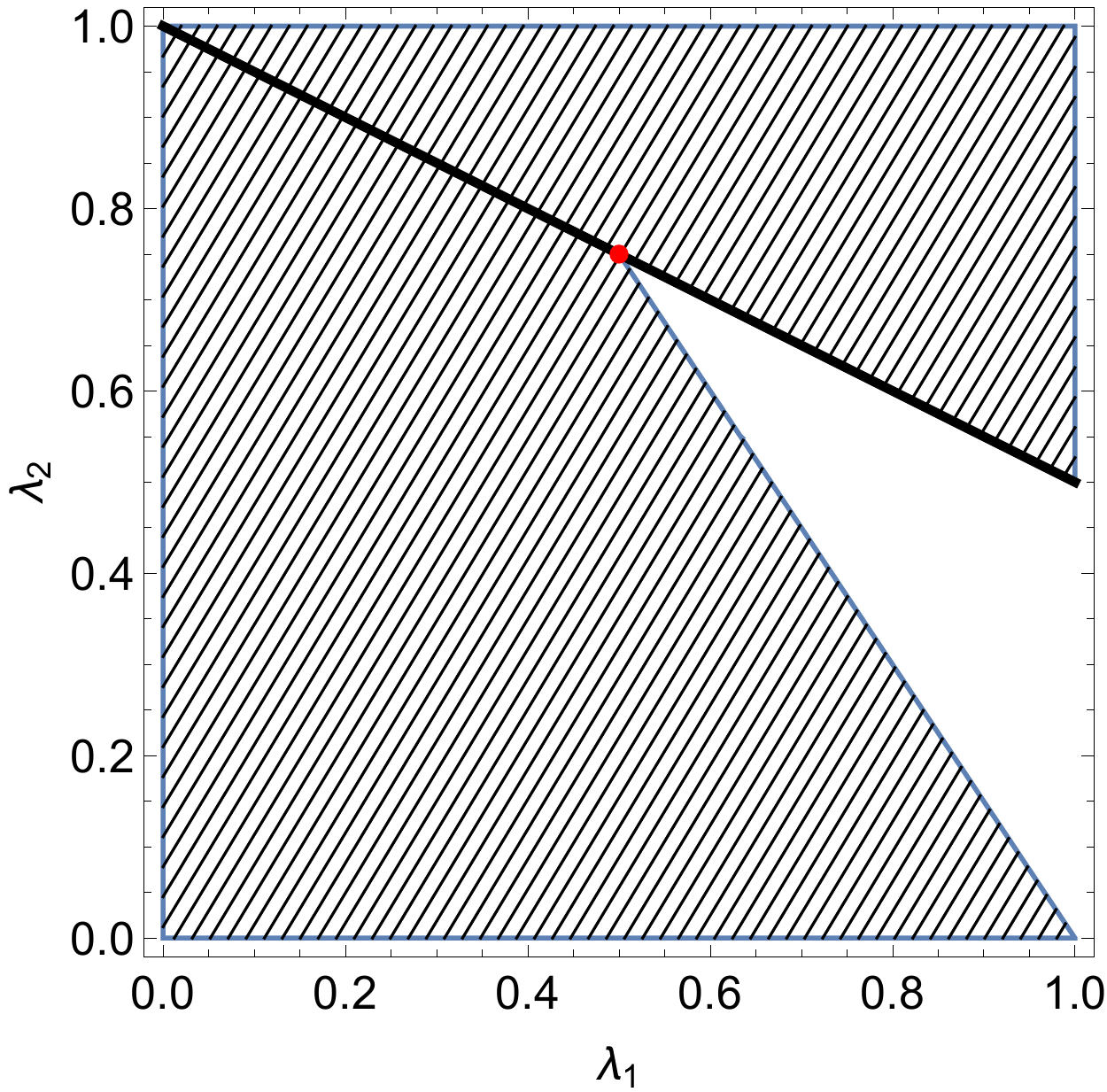}
\caption{Shaded region: unreachable parameters given a strategy $\Omega$ that is both local and of the form $\Omega = \alpha P^+_{ZZ} + (1-\alpha)\Omega_3$, where $\Omega_3$ is the trace 3 part. Here, $\theta = \frac{\pi}{8}$.}
\label{fig:convhull}
\end{minipage}\hfill%
\begin{minipage}[t]{0.48\linewidth}
\includegraphics[width=\linewidth]{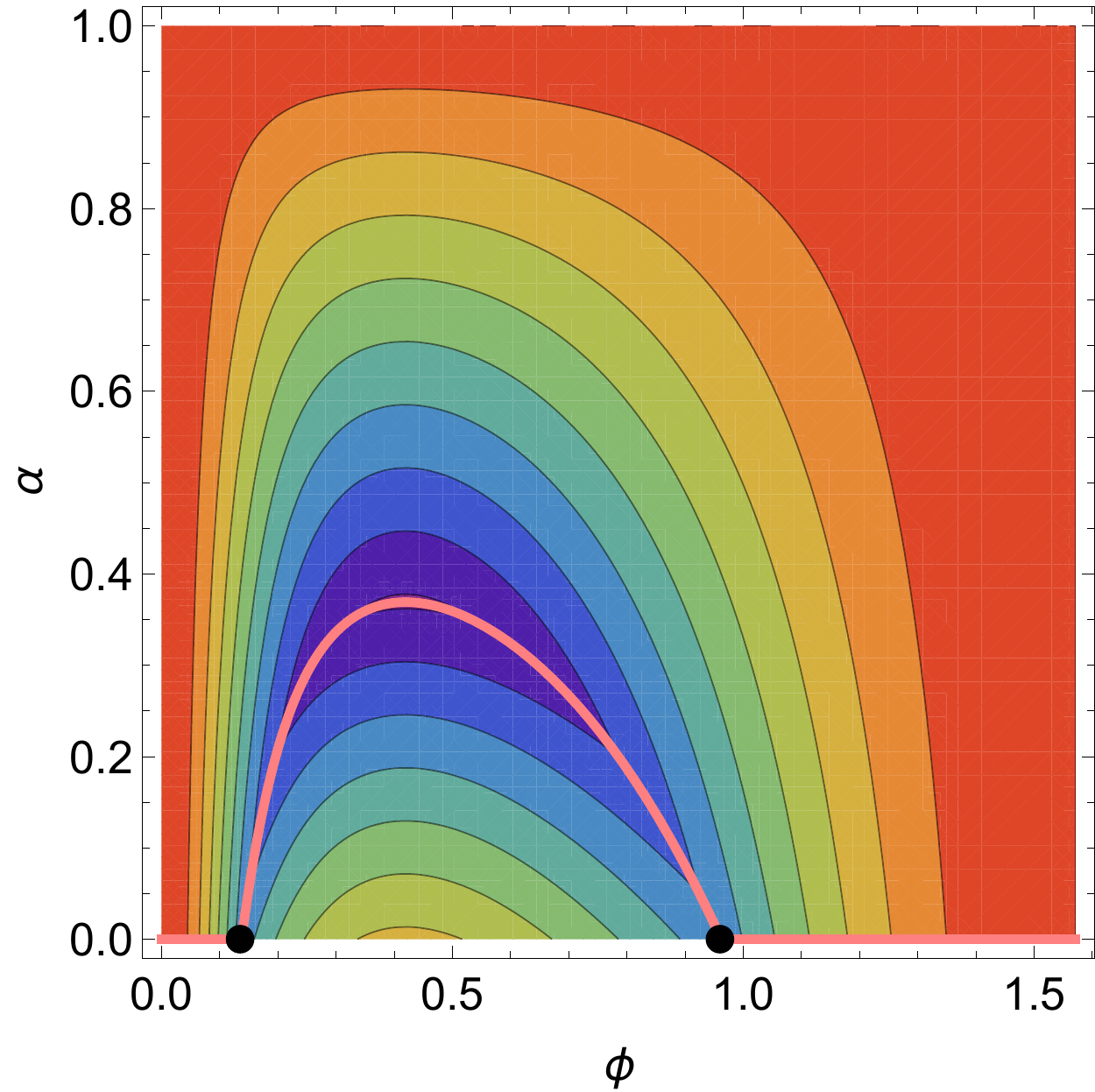}
\caption{A contour map of the function $q(\alpha,\phi)=\max\{\lambda_1(\alpha,\phi),\lambda_2(\alpha,\phi)\}$ for $\theta = \frac{\pi}{8}$, where the pair $(\lambda_1,\lambda_2)$ are given in~\ref{eq:strat_eigenvalues}. The pink curve denotes the minimum w.r.t $\alpha$ given fixed $\phi$. Above the curve, $\lambda_1 > \lambda_2$; below, $\lambda_1 < \lambda_2$.}
\label{fig:alpha_phi_opt}
\end{minipage}%
\end{figure*}

\noindent We now briefly treat the special cases that were omitted from the above proof: $\theta = 0, \frac{\pi}{4}, \frac{\pi}{2}$.

\noindent $\mathit{\theta = 0, \theta = \frac{\pi}{2}}$: In these cases, the state $\ket{\psi} = \ket{00}$ or $\ket{\psi}=\ket{11}$. Then the globally optimal strategy, just projecting onto $\ket{\psi}$, is an allowed local measurement. Thus in these cases the optimal strategy is to just apply the projector $\ket{00}\bra{00}$ or $\ket{11}\bra{11}$. Given this strategy we have that $p=1$ and $q=0$, giving a scaling of the number of measurements required as
\begin{equation}
n_{opt} \approx \epsilon^{-1}\ln\delta^{-1}.
\end{equation}
$\mathit{\theta = \frac{\pi}{4}}$: This case is treated explicitly in the main body. The optimal strategy is to perform the Pauli measurements $XX$, $-YY$ and $ZZ$ with equal weight; i.e.
\begin{equation}
\Omega = \frac{1}{3}(P^+_{XX}+P^+_{-YY}+P^+_{ZZ}),
\end{equation}
where $P^+_{M}$ is the projector onto the positive eigensubspace of the operator $M$. In this case, the number of measurements required is
\begin{equation}
n_{opt} \approx \frac{3}{2}\epsilon^{-1}\ln\delta^{-1}.
\end{equation}


\section{Stabilizer states}\label{app:stabilizers}

We now discuss verification strategies for stabilizer states. We take $\ket{\psi}$ to be a stabilizer state of $N$ qubits, namely that there exists a generating set of $N$ commuting Pauli operators $M_1,\dots,M_N$ on $N$ qubits such that $M_i \ket{\psi} = \ket{\psi}$ for all $i$. Stabilizer states are ubiquitous in various areas of quantum information, for example in quantum error correction and measurement-based quantum computing; for an introduction to the stabilizer formalism, see~\cite{Gottesman1997,Gottesman1996} and~\cite{Nielsen2010} Sec 10.5. We will describe below a strategy constructed from only stabilizer measurements that accepts $\ket{\psi}$ with certainty, and hence achieves the same asymptotic scaling in the number of required measurements with respect to $\epsilon$ as the two-qubit case above. However, we do not rule out that there may be non-stabilizer strategies that give a small constant factor improvement over the strategy defined here.

\begin{thm}
Write a stabilizer state $\ket{\psi}$ and strategy $\Omega = \sum_{j=1}^K \mu_j P_j$, where the set $\{P_j\}$ are the projectors onto the positive eigenspace of $K$ linearly independent stabilizers of $\ket{\psi}$, for $K \leq 2^N - 1$. Then the optimal choice of the parameter $K$ and weights $\mu_j$ are $K = 2^N - 1; \; \mu_j = \frac{1}{2^N - 1}$ for all $j$. The number of measurements needed to verify to within fidelity $\epsilon$ and statistical power $1-\delta$ is then
\begin{equation}
n_{opt}^{stab} \approx \frac{2^N - 1}{2^{(N-1)}} \epsilon^{-1}\ln\frac{1}{\delta}.
\end{equation}
\end{thm}

\begin{proof}
Recall that as the verifier accepts $\ket{\psi}$ with certainty, we are concerned with the optimisation of $\Delta_\epsilon$, which can be written as
\begin{align}
\Delta_\epsilon &= \max_\Omega \min_{\ket{\psi^\bot}} \epsilon(1-\bra{\psi^\bot}\Omega\ket{\psi^\bot})\\
&= \epsilon(1- \min_\Omega \max_{\ket{\psi^\bot}} \bra{\psi^\bot}\Omega\ket{\psi^\bot}),
\end{align}
where the maximisation is over positive matrices $\Omega$ such that $\Omega\ket{\psi} = \ket{\psi}$.

Now consider $\Omega$ written as a matrix in the basis $\{\ket{\psi},\ket{\psi_j^\bot}\}$, $j = 1 \ldots (2^N - 1)$ where the states $\ket{\psi_j^\bot}$ are mutually orthogonal and all orthogonal to $\ket{\psi}$. Given that $\Omega\ket{\psi} = \ket{\psi}$, we know that $\bra{\psi_j^\bot}\Omega\ket{\psi}=0 \; \forall j$. Then in this basis $\Omega$ can be written 
\begin{equation}
\Omega = \begin{pmatrix}1&\mathbf{0}^\top\\\mathbf{0}&\mathbf{M}\end{pmatrix},
\end{equation}
where $\mathbf{0}$ is the $(2^N - 1)$-dimensional zero vector and $\mathbf{M}$ is a $(2^N-1)\times(2^N-1)$ Hermitian matrix. Then $\Omega$ must be writable as $\Omega = \ket{\psi}\bra{\psi} + \sum_{j=1}^{2^N - 1} \nu_j \ket{\phi_j}\bra{\phi_j}$, where $\sum_j \nu_j \ket{\phi_j}\bra{\phi_j}$ is the spectral decomposition of $\mathbf{M}$. Given this decomposition, the optimisation for the adversary is straightforward -- pick $\ket{\psi^\bot}$ to be the eigenstate in the decomposition of $\mathbf{M}$ with largest eigenvalue: $\ket{\psi^\bot} = \ket{\phi_{max}}$ where $\nu_{max} = \max_j \nu_j$. Then
\begin{equation}
\Delta_\epsilon = \epsilon(1-\min_\Omega \bra{\phi_{max}}\Omega\ket{\phi_{max}}) = \epsilon(1-\min_\Omega \nu_{max}).
\end{equation}
Given this choice by the adversary, the verifier is then forced to set the strategy such that all the eigenvalues of $\mathbf{M}$ are equal; i.e. that $\mathbf{M} = a\Id$ for some constant $a$. To see this, consider an alternative strategy where the eigenvalues $\nu_j$ are not equal. Now, consider rewriting $\Omega$ in terms of stabilizers of $\ket{\psi}$. For any stabilizer (i.e.\ tensor product of Paulis, perhaps with an overall phase) $M$ over $N$ qubits, the projector onto the positive eigensubspace has $\tr(P_M^+) = 2^{N-1}$. Given that $\Omega$ is built from a convex combination of these projectors, and recalling from Lemma \ref{lem:no_id_part} that $\Omega$ does not contain an identity term, we also know that $\tr(\Omega) = 2^{N-1}$. However, we have also expanded $\Omega$ as $\Omega = \ket{\psi}\bra{\psi} + \sum_j \nu_j \ket{\phi_j}\bra{\phi_j}$, and so
\begin{equation}
\tr(\Omega) = 1 + \sum_j \nu_j = 2^{N-1}.
\end{equation}
Then, it is straightforward to see that decreasing any eigenvalue below $a$ must result in an increase in at least one other eigenvalue in order to maintain this equality, and hence would increase the value of $\nu_{max}$. Thus the optimal choice for the verifier is to set $\Omega = \ket{\psi}\bra{\psi} + a\Id^\bot$, where $\Id^\bot$ is the identity matrix on the subspace orthogonal to $\ket{\psi}$. Taking the trace of this expression gives
\begin{equation}
    \tr [\ket{\psi}\bra{\psi} + a\Id^\bot] = 1 + (2^N - 1)a = 2^{N-1}.
\end{equation}
This can be rearranged for $a$ and then substituted into the expression for $\Delta_\epsilon$, which gives
\begin{equation}
\Delta_\epsilon = \frac{2^{(N-1)}}{2^N - 1}\epsilon,
\end{equation}
or that the number of stabilizer measurements required to verify $\ket{\psi}$ is bounded below by
\begin{equation}\label{eq:stab_n}
n_{opt}^{stab} \approx \frac{2^N - 1}{2^{(N-1)}} \epsilon^{-1}\ln\delta^{-1}.
\end{equation}
The optimal $\Omega = \ket{\psi}\bra{\psi} + \frac{2^{(N-1)}-1}{2^N - 1}\Id^\bot$ and the optimal scaling above can be achieved by decomposing $\Omega$ into a strategy involving a maximal set (excluding the identity) of $2^N - 1$ linearly independent stabilizers, all with equal weight. To see this note that for a stabilizer group of a state $\ket{\psi}$ of $N$ qubits, there are $2^N$ linearly independent stabilizers (including the identity element). Denote these stabilizers $\{M_i, i=1 \ldots 2^N\}$. Then, we make use of the fact that \cite{Hein2005}

\begin{equation}
\frac{1}{2^N} \sum_{i=1}^{2^N} M_i = \ket{\psi}\bra{\psi}.
\end{equation}
Explicitly extracting the identity element gives
\begin{equation}
\sum_{i=1}^{2^N-1} M_i = 2^N \ket{\psi}\bra{\psi} - \Id.
\end{equation}
Now, each stabilizer (for any $N$) is a two outcome measurement and so we can make use of the fact that $M_i$ can be written in terms of the projector onto the positive eigenspace of $M_i$, denoted $P^+_i$, as $M_i = 2P^+_i - \Id$. Substituting in this expression and rearranging gives
\begin{equation}
\sum_{i=1}^{2^N - 1} P^+_i = 2^{(N-1)}\ket{\psi}\bra{\psi} + (2^{(N-1)}-1)\Id.
\end{equation}
Then normalising this expression over $2^N - 1$ stabilizers yields
\begin{align}
\nonumber \frac{1}{2^N - 1}\sum_{i=1}^{2^N - 1}P^+_i &= \frac{2^{(N-1)}}{2^N - 1}\ket{\psi}\bra{\psi} + \frac{2^{(N-1)}-1}{2^N - 1}\Id\\
\nonumber &= \frac{2^{(N-1)}+2^{(N-1)}-1}{2^N - 1}\ket{\psi}\bra{\psi} + \frac{2^{(N-1)}-1}{2^N - 1}\Id^\bot\\
&= \ket{\psi}\bra{\psi} + \frac{2^{(N-1)}-1}{2^N - 1}\Id^\bot = \Omega,
\end{align}
where $\Id^\bot$ is the identity matrix on the subspace orthogonal to $\ket{\psi}$, as required.
\end{proof}

Note that for growing $N$, the quantity $n_{opt}^{stab}$ given in Eq.~\ref{eq:stab_n} is bounded above by $2\epsilon^{-1}\ln \delta^{-1}$, which does not depend on $N$, and implies that this stabilizer strategy requires at most a factor of two more measurements than the optimal non-local verification strategy (just projecting onto $\ket{\psi}$).

One could also consider a reduced strategy that involves measuring fewer stabilizers. However, given a state of $N$ qubits and a set of $k$ stabilizers, the dimension of the subspace stabilized by this set is at least $2^{N-k}$. Thus for any choice of $k<N$, there must always exist at least one state $\ket{\psi^\bot}$ orthogonal to $\ket{\psi}$ that is stabilized by every stabilizer in the set. Then, the adversary can construct a $\sigma$ that always accepts, implying that the verifier has no discriminatory power between $\ket{\psi}$ and $\sigma$ and thus the strategy fails. Consider instead constructing a strategy from the $N$ stabilizer generators of $\ket{\psi}$, with corresponding projectors $\{P^{s.g.}_j\}$. Then, $\Omega = \sum_j \mu_j P^{s.g.}_j$. The set of projectors $\{P^{s.g.}_j\}$ commute and so share a common eigenbasis, denoted $\{\ket{\lambda_j}\}$. To optimise this strategy over the weights $\mu_j$, we first need the following lemma:
\begin{lem}\label{lem:sg}
Write the unique sets of $N-1$ independent stabilizer generators of $\ket{\psi}$, $S_k = \{M_j, j=1\ldots N\} \setminus M_k$, $k = 1 \ldots N$. Then each $S_k$ corresponds to a state $\ket{\lambda_k}$, $\braket{\lambda_k|\psi}=0$, such that $\braket{\lambda_k|\lambda_l}=\delta_{kl}$.
\end{lem}
\begin{proof}
Each set $S_k$ stabilizes a space of dimension two, and so a $\ket{\lambda_k}$ where $\braket{\lambda_k|\psi}=0$ exists. Moreover, the stabilizer generators define an orthogonal eigenbasis of which $\ket{\lambda_k}$ is an element. To show that two sets $S_k$ and $S_l$, $k \neq l$, define distinct eigenvectors, assume the converse; that $\ket{\lambda_k} \propto \ket{\lambda_l}$. However, then the set $S = S_k \cup S_l$ would stabilize $\ket{\lambda_k}$, which is a contradiction as $S$ is the full set of stabilizer generators and uniquely stabilizes $\ket{\psi}$.
\end{proof}
\noindent We can now derive the optimal stabilizer generator strategy.
\begin{thm}
For a stabilizer state $\ket{\psi}$ and strategy $\Omega = \sum_{j=1}^N \mu_j P^{s.g.}_j$, where the set $\{P^{s.g.}_j\}$ are the projectors onto the positive eigenspace of the stabilizer generators of $\ket{\psi}$, the optimal choice of the weights $\mu_j$ is $\mu_j = \frac{1}{N}$, for all $j$. The number of measurements needed to verify to within fidelity $\epsilon$ and statistical power $1-\delta$ is then
\begin{equation}
n^{s.g.}_{opt} \approx \frac{N}{\epsilon}\ln\frac{1}{\delta}.
\end{equation}
\end{thm}
\begin{proof}
If we write a state orthogonal to $\ket{\psi}$ in the stabilizer eigenbasis as $\ket{\psi^\bot} = \sum_k \alpha_k \ket{\lambda_k}$, we have that
\begin{align}
\nonumber \bra{\psi^\bot}\Omega\ket{\psi^\bot} &= \sum_{k,m = 1}^{2^N}\sum_{j=1}^{N} \bar{\alpha}_k \alpha_m \mu_j \bra{\lambda_k} P^{s.g.}_j\ket{\lambda_m}\\
\nonumber &= \sum_{k,m = 1}^{2^N}\sum_{j=1}^{N} \bar{\alpha}_k \alpha_m \mu_j \delta_{km}\epsilon_{jk}\\
&= \sum_{k = 1}^{2^N}\sum_{j=1}^{N} \vert \alpha_k \vert^2 \mu_j \epsilon_{jk} \coloneqq \sum_{k=1}^{2^N} \vert \alpha_k \vert^2 E_k,
\end{align}
where $\epsilon_{jk} = 1$ if $P_j\ket{\lambda_k}=\ket{\lambda_k}$ and zero otherwise. This quantity is the \emph{parity-check matrix} for the set of stabilizers $\{P_j^{s.g.}\}$. The quantity of interest with respect to verification is
\begin{equation}
q = \min_{\Omega}\max_{\ket{\psi^\bot}} \bra{\psi^\bot}\Omega\ket{\psi^\bot} = \min_{\mu_j}\max_{\alpha_k} \sum_{j,k} \vert \alpha_k \vert^2 \mu_j \epsilon_{jk},
\end{equation}
where the verifier's minimisation is over the probabilities $\mu_j$ with which a stabilizer generator indexed by $j$ is drawn in the protocol, and the adversary maximises over the set of amplitudes $\alpha_k$ that describes the state most likely to fool the verifier. Lemma~\ref{lem:sg} gives that, from the full set of $2^N$ basis states $\ket{\lambda_k}$, there is a subset of $N$ basis states $\ket{\lambda_{\tilde{k}}}$, $\tilde{k} \in I$ for $\vert I \vert = N$, stabilized by exactly $N-1$ generators; thus for basis states in this subset, the quantity $\epsilon_{j\tilde{k}}=1-\delta_{j\tilde{k}}$. Then we can compute the summation over $j$ as
\begin{equation}
E_{\tilde{k}} = \sum_j \mu_j \epsilon_{j\tilde{k}}= \sum_j \mu_j (1-\delta_{j\tilde{k}}) = 1 - \mu_{\tilde{k}},
\end{equation}
using the fact that $\sum_j \mu_j = 1$. Now, each element of $E_k$ for $k \notin I$ is a summation of at most $N-2$ terms, $\mu_j$. Thus there always exists another element $E_{\tilde{k}}$ for $\tilde{k} \in I$ that is at least as large; and so it is never detrimental to the adversary to shift any amplitude on the basis state labelled by $k$ to the basis state labelled by $\tilde{k}$. Thus the optimal choice for the adversary's state is $\ket{\psi^\bot} \in \text{span}\{\ket{\lambda_{\tilde{k}}}: \tilde{k} \in I\}$. Given this choice by the adversary, we have that
\begin{equation}
q = \min_{\mu_{\tilde{k}}}\max_{\alpha_{\tilde{k}}} \sum_{\tilde{k}} \vert \alpha_{\tilde{k}}\vert^2 (1-\mu_{\tilde{k}}) = \min_{\mu_{\tilde{k}}}\max_{\tilde{k}}(1-\mu_{\tilde{k}}).
\end{equation}
It is straightforward to see that the optimal choice for the verifier is to have $\mu_{\tilde{k}} = \frac{1}{N}$, for all $\tilde{k}$; then $\Omega = \frac{1}{N}\sum{P_j^{s.g.}}$. Thus
\begin{equation}
q = 1 - \frac{1}{N} \Rightarrow n^{s.g.}_{opt} \approx \frac{N}{\epsilon}\ln\frac{1}{\delta}.
\end{equation}
\end{proof}

Clearly, this scaling is much poorer in $N$ than in the case where the full set of $2^N - 1$ linearly independent stabilizers are allowed; indicating a trade-off between the total number of required measurements and the accessible number of measurement settings, in this case.

\section{Concentration inequalities and the relative entropy}\label{app:hypothesis_testing}
\noindent 
In a binary hypothesis test between hypotheses $H_0$ and $H_1$, the Type I and Type II errors are, respectively,
\begin{alignat}{2}
\text{Type I}&: \quad &&\text{Pr}[\text{Guess } H_1 \vert H_0] \\
\text{Type II}&: &&\text{Pr}[\text{Guess } H_0 \vert H_1].
\end{alignat}

\noindent In general, in designing an effective hypothesis test there will be a trade-off between the relative magnitude of these types of error; they cannot be arbitrarily decreased simultaneously. In an \emph{asymmetric} hypothesis test, the goal is to minimise one of these errors given a fixed upper bound on the other. In this addendum, we prove the following proposition in the context of asymmetric hypothesis testing:

\begin{prop}\label{prop:p=1}
Any strategy $\Omega$ that: (a) accepts $\ket{\psi}$ with certainty, $p \coloneqq \tr(\Omega \ket{\psi}\bra{\psi}) = 1$; and (b) does not accept $\sigma$ with certainty ($\Delta_\epsilon>0$) requires asymptotically fewer measurements in infidelity $\epsilon$ to distinguish these states to within a fixed Type II error than the best protocol based on a strategy $\Omega'$ where $\tr(\Omega' \ket{\psi}\bra{\psi}) <1$.
\end{prop}

We have inherited notation regarding verification strategies from Appendix~\ref{app:verification}. Here, hypothesis $H_0$ corresponds to accepting the target $\ket{\psi}$, and hypothesis $H_1$ corresponds to accepting the alternative (that the output was far from $\ket{\psi}$). Proposition \ref{prop:p=1} states that, in a framework where we attempt to verify $\ket{\psi}$ by repeatedly making two-outcome measurements picked from some set, asymptotically it is always beneficial to use measurements that accept $\ket{\psi}$ with certainty. In this case, each measurement is a Bernoulli trial with some acceptance probability. An example of a protocol which would {\em not} satisfy this property would be estimating the probability of violating a Bell inequality for a maximally entangled 2-qubit state $\ket{\psi}$.

In general, the optimum asymptotic rate at which the Type II error can be minimised in an asymmetric hypothesis test is given by the \emph{Chernoff-Stein lemma}:
\begin{thm}[Cover and Thomas~\cite{Cover2006}, Theorem 11.8.3.]\label{thm:chernoff_stein}
Let $X_1, X_2 \ldots X_n$ be drawn i.i.d. from a probability mass function $Q$. Then consider the hypothesis test between alternatives $H_0$: $Q=P_0$ and $H_1$: $Q=P_1$. Let $A_n$ be an acceptance region for the null hypothesis $H_0$; i.e. it is a set consisting of all possible strings of outcomes with which the conclusion $H_0$ is drawn. Denote Type I and Type II errors after $n$ samples as $\alpha_n^*$ and $\beta_n^*$, respectively. 
Then for some constraint parameter $0 < \chi < \frac{1}{2}$, define
\begin{equation*}
\delta_n^\chi = \min_{\substack{A_n \\ \alpha_n^* < \chi}} \beta_n^*.
\end{equation*}
\noindent Then asymptotically
\begin{equation*}
\lim_{n \rightarrow \infty} \frac{1}{n}\ln\delta_n^\chi = -D(P_0\ \Vert P_1),
\end{equation*}
where $D(P_0 \Vert P_1)$ is the relative entropy between probability distributions $P_0$ and $P_1$.
\end{thm}

For clarity we drop the sub- and superscript $\delta_n^\chi \rightarrow \delta$. The relative entropy typically takes a pair of probability distributions as arguments, but given that each hypothesis is concerned only with a single Bernoulli-distributed random variable uniquely specified by a a pair of real parameters (the quantities $p$ and $p-\Delta_\epsilon$), we will use the shorthand $D(p \Vert q)$ for real variables $p$ and $q$. In this case the relative entropy can be expanded as
\begin{equation}
    D(a \Vert b) = a\ln\frac{a}{b} + (1-a)\ln\frac{1-a}{1-b}.
\end{equation}
Note that in the limit where $a\rightarrow 1$, using that $\lim_{a\rightarrow 1^-} (1-a)\ln(1-a) = 0$, this expression becomes
\begin{equation}
\lim_{a\rightarrow 1^-} D(a \Vert b) = \ln\frac{1}{b}.
\end{equation}
After rearranging the expression for the optimal asymptotic Type II error given by the Chernoff-Stein lemma, we can achieve a test with statistical power $1-\delta$ by taking a number of measurements
\begin{equation}
    n > \frac{1}{D\left(p\Vert p -\Delta_{\epsilon}\right)}\ln \frac{1}{\delta}.
\end{equation}
Moreover, this bound is tight in that it gives the correct asymptotic relationship between $n$, $D$ and $\delta$; generically $\delta$ can be lower bounded (\cite{Cover2006}, p666) such that
\begin{equation}
    \frac{e^{-Dn}}{n+1} \leq \delta \leq e^{-Dn}.
\end{equation}

Two important limiting cases of this expression have relevance here. Firstly, if $p \gg \Delta_\epsilon$, then Taylor expanding $n$ for small $\Delta_\epsilon$ gives that it is sufficient to take
\begin{equation}
    n \geq \frac{2 p(1-p)}{\Delta_\epsilon^2}\ln\frac{1}{\delta}.
\end{equation}
Secondly, if $p = 1$, then it is sufficient to take
\begin{equation}
n \geq \frac{-1}{\ln\left(1-\Delta_\epsilon\right)}\ln\frac{1}{\delta} \approx \frac{1}{\Delta_\epsilon}\ln\frac{1}{\delta},
\end{equation}
which is in agreement with the scaling previously derived in Eq.~\ref{eq:globalopt}. These are the limiting cases of the scaling of $n$ with $\Delta_\epsilon$. In the worst case, $n$ scales quadratically in $\Delta_\epsilon^{-1}$; however, for any strategy where the state $\ket{\psi}$ to be tested is accepted with certainty, only a total number of measurements linear in $\Delta_\epsilon^{-1}$ are required. Thus asymptotically, a strategy where $p=1$ is always favourable (i.e. gives a quadratic improvement in scaling with $\Delta_\epsilon$) for any $\Delta_\epsilon > 0$.

\bibliographystyle{apsrev4-1}
\bibliography{Verification}

\end{document}